\documentclass[aps,pra,twocolumn,superscriptaddress,nofootinbib,longbibliography]{revtex4-2}

\usepackage{amsmath, empheq,braket,amsthm,amssymb,graphics,amsfonts,array,color,subfigure, mathrsfs, dsfont, physics, float}
\usepackage[normalem]{ulem}

\usepackage{comment}

\definecolor{myurlcolor}{rgb}{0,0,0.7}
\definecolor{myurlcolor1}{rgb}{0,0.7,0.1}
\definecolor{myrefcolor}{rgb}{0,0,0.7}
\usepackage[unicode=true,pdfusetitle, bookmarks=true,bookmarksnumbered=false,bookmarksopen=false, breaklinks=false,pdfborder={0 0 0},backref=false,colorlinks=true, linkcolor=myrefcolor,citecolor=myurlcolor1,urlcolor=myurlcolor]{hyperref}

\usepackage{tensor} 
\usepackage{scalerel} 

\usepackage{url} 

\usepackage{hyperref}

\newcommand{\kb}[2]{\left| #1 \vphantom{#2} \right>\left< #2 \vphantom{#1} \right|} 
\newcommand{\proj}[1]{\kb{#1}{#1}} 

\renewcommand{\tr}[1]{\mathrm{Tr} ( #1 )}

\renewcommand{\dd}{\mathrm{d}}

\DeclarePairedDelimiterX{\infdivx}[2]{(}{)}{%
  #1\;\delimsize\|\;#2%
}

\newcommand{\Vol}{\mathrm{Vol}}

\newtheorem{prop}{Property}

	 \newcommand{\eq}[1]{\begin{equation}#1	\end{equation}}
	 
\newcommand{\red}[1]{\textcolor{red}{#1}}

\newcommand{\p}{\hat{p}_1}

\usepackage{bm}

\begin{document}

\title{Complex-valued Wigner entropy of a quantum state}

\author{Nicolas J. Cerf}
\affiliation{Centre for Quantum Information and Communication, \'{E}cole polytechnique de Bruxelles,  CP 165, Universit\'{e} libre de Bruxelles, 1050 Brussels, Belgium}
\affiliation{James C. Wyant College of Optical Sciences, University of Arizona, Tucson, AZ 85721, USA}

\author{Anaelle Hertz}
\affiliation{National Research Council of Canada, 100 Sussex Drive, Ottawa, ON K1N 5A2, Canada}

\affiliation{Centre for Quantum Information and Communication, \'{E}cole polytechnique de Bruxelles,  CP 165, Universit\'{e} libre de Bruxelles, 1050 Brussels, Belgium}

\author{Zacharie Van Herstraeten}
\affiliation{James C. Wyant College of Optical Sciences, University of Arizona, Tucson, AZ 85721, USA}
\affiliation{Centre for Quantum Information and Communication, \'{E}cole polytechnique de Bruxelles,  CP 165, Universit\'{e} libre de Bruxelles, 1050 Brussels, Belgium}

\begin{abstract}
It is common knowledge that the Wigner function of a quantum state may admit negative values, so that it cannot be viewed as a genuine probability density. Here, we examine the difficulty in finding an entropy-like functional in phase space that extends to negative Wigner functions and then advocate the merits of defining a complex-valued entropy associated with any Wigner function. This quantity, which we call the \textit{complex Wigner entropy}, is defined via the analytic continuation of Shannon's differential entropy of the Wigner function in the complex plane. We show that the complex Wigner entropy enjoys interesting properties, especially its real and imaginary parts are both invariant under Gaussian unitaries (displacements, rotations, and squeezing in phase space). Its real part is physically relevant when considering the evolution of the Wigner function under a Gaussian convolution, while its imaginary part is simply proportional to the negative volume of the Wigner function. Finally, we define the complex-valued Fisher information of any Wigner function, which is linked (via an extended de~Bruijn's identity) to the time derivative of the complex Wigner entropy when the state undergoes Gaussian additive noise. Overall, it is anticipated that the complex plane yields a proper framework for analyzing the entropic properties of quasiprobability distributions in phase space. 
\end{abstract}

\maketitle

\section{Introduction}
When representing a quantum state with a quasiprobability distribution in phase space, the most common choices are the Glauber-Sudarshan P\mbox{-}function $P(\alpha)$, the Wigner function $W(\alpha)$, or the Husimi Q\mbox{-}function $Q(\alpha)$, where $\alpha$ is a complex variable \cite{LEE1995147}. These correspond to the normal, symmetric, and antinormal ordering of annihilation  and creation operators, respectively, and belong to a continuous family of $s$-ordered quasiprobability distributions $W_s(\alpha)$ parametrized by the real parameter $s\in [-1,1]$, see \cite{PhysRev.177.1857,PhysRev.177.1882}. Conventionally, $W_1(\alpha)=P(\alpha)$, $W_0(\alpha)=W(\alpha)$, and $W_{-1}(\alpha)=Q(\alpha)$. 
The P-function (\textit{i.e.}, $s=1$) of a state  is often not even a regular function (for instance, it is expressed in terms of derivatives of the Dirac $\delta$\mbox{-}function for a squeezed vacuum state). However, by decreasing the value of $s$ (which corresponds to convolving $P(\alpha)$ with a Gaussian distribution of increasing variance), we reach a point where $W_s(\alpha)$ becomes regular, see example in Fig.~\ref{Fig:intro}. The largest value of $s$ for which $W_s(\alpha)$ is a regular function for \emph{all} states is $s=0$, corresponding to the Wigner function $W(\alpha)$. The Wigner function is thus always regular but, as is well known, it is often not a positive function, so it cannot be treated as a genuine probability distribution. 

\begin{figure}[t!]
\centering
\includegraphics[width=\linewidth]{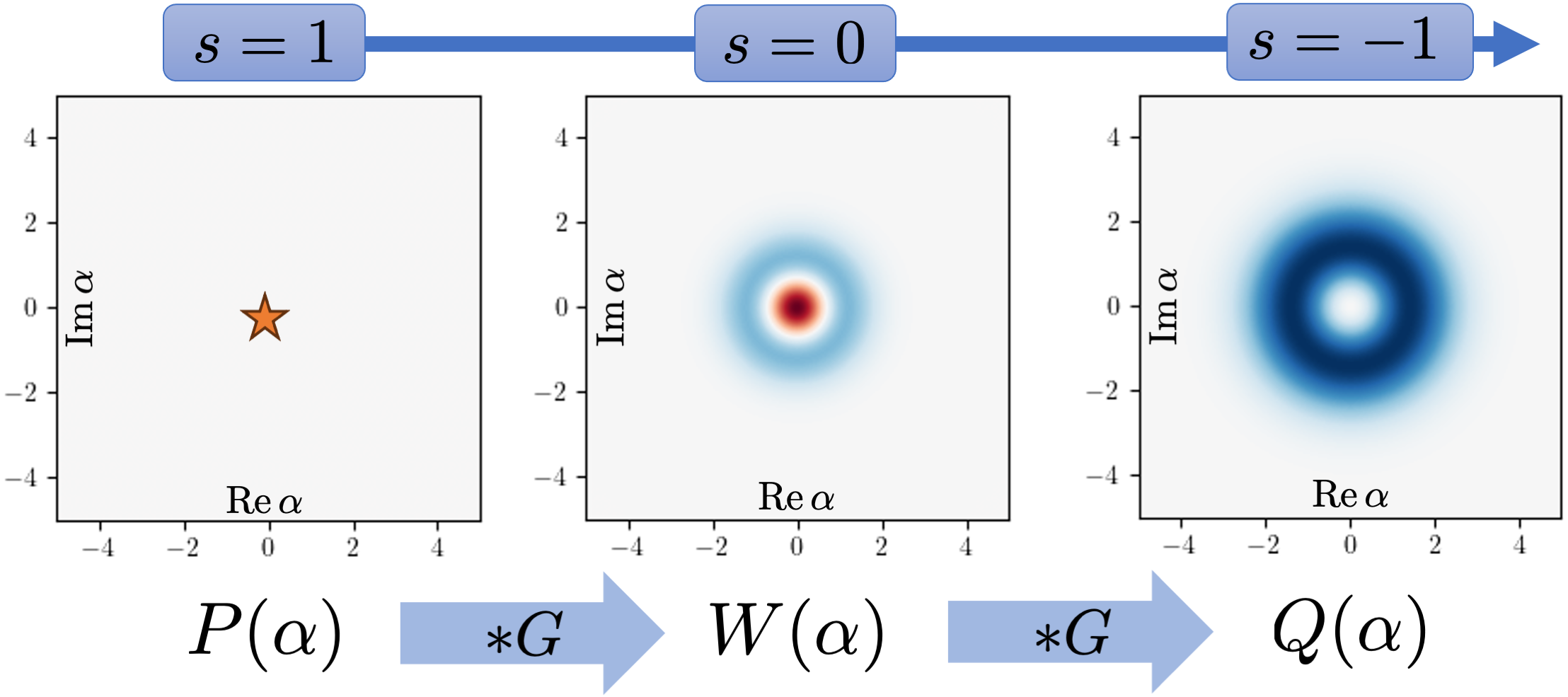}
\caption{
Phase-space representations of the Fock state $\ket{1}$.
Its Glauber-Sudarshan P-function $P(\alpha)$ is expressed in terms of the second derivative of a Dirac $\delta$-function (represented by a red star) and is thus not regular.
Smoothing $P(\alpha)$ with a Gaussian distribution of one unit of shot-noise (noted as the convolution $*G$) yields the Wigner function $W(\alpha)$, which is regular but takes negative values (red zone) near the origin.
Smoothing again the Wigner function $W(\alpha)$ with the same Gaussian distribution yields the Husimi function $Q(\alpha)$, which is both regular and non-negative.
Here, we focus on defining the complex-valued entropy of $W(\alpha)$ instead of the real-valued entropy of $Q(\alpha)$, which is the well-known Wehrl entropy.}
\label{Fig:intro}
\end{figure}

The negativities of $W(\alpha)$ can be suppressed by further decreasing $s$ below 0, which smears out the quasiprobability distribution up to the point where it becomes positive everywhere in phase space. The largest value of $s$ for which $W_s(\alpha)$ is a non-negative function for \emph{all} states is $s=-1$, corresponding to the Q-function $Q(\alpha)$, see example in Fig.~\ref{Fig:intro}. This is the only quasiprobability distribution among $W_s(\alpha)$ with $s\in[-1,1]$ that behaves as a genuine probability distribution (this is because it is the probability distribution of the outcome of a POVM measurement in the overcomplete basis of coherent states $\ket{\alpha}$, often called heterodyne or eight-port homodyne detection in quantum optics).

Defining a sensible measure of uncertainty based on the representation of a state in quantum phase space is an objective that has long been pursued, and a natural candidate for this is the differential entropy as defined in Shannon information theory \cite{Shannon1948,ShannonWeaver}. Notably, the Shannon differential entropy of the Q-function associated with a state is called the Wehrl entropy \cite{RevModPhys.50.221} and it has extensively been analyzed in the literature. For example, the Wehrl conjecture \cite{WEHRL1979353} (later proven by Lieb \cite{Lieb1978}) gives a lower bound to the Wehrl entropy of any state. Let us stress here that the Q-function is the only one quasiprobability distribution $W_s(\alpha)$ with $s\in[-1,1]$ whose Shannon differential entropy is always well defined for \textit{all} states as the entropy definition would otherwise require taking the logarithm of negative numbers for some states.

Here, we revisit the definition of the entropy of a quantum phase-space distribution $W_s(\alpha)$ by relaxing the strong requirement of having a genuine probability distribution, yet requiring the function $W_s(\alpha)$ to remain regular for all states. The most natural choice is then the Wigner function $W(\alpha)$, which also enjoys numerous interesting properties and is by far the most commonly used phase-space representation in quantum optics \cite{mandel_wolf_1995}. Among its main properties, let us mention the phase-space area preservation under Gaussian unitaries, which does not hold for any other $W_s(\alpha)$ with $s\ne 0$. Furthermore, its marginals yield the probability distributions of the $x$ and $p$ canonical coordinates, so it can be thought of as a sort of joint $(x,p)$ probability distribution albeit violating Kolmogorov's non-negativity axiom.

It thus makes sense to define the Wigner entropy \cite{PhysRevA.104.042211} of a state as the differential Shannon entropy of $W(\alpha)$, which is 
 well behaved as long as one restricts to Wigner-positive states, \textit{i.e.}, states such that $W(\alpha)\ge 0$, $\forall \alpha\in\mathbb{C}$. A great merit of this definition is that it inherits all nice properties of Wigner functions. For example, the Wigner entropy is invariant under Gaussian unitaries as a consequence of the area-preservation property (notice that such an invariance does not hold for the Wehrl entropy, nor for any other $W_s(\alpha)$ with $s<0$).  In fact, as long as the considered states are Wigner-positive, all results of Shannon information theory remain applicable to $W(\alpha)$ since the latter is a true probability distribution (see, \textit{e.g.}, the Wigner entropy-power inequality \cite{PhysRevA.104.042211} for a beam splitter). The Wigner entropy is thus also a lower bound to the sum of marginal entropies of $x$ and $p$ as a consequence of the subadditivity of Shannon entropy, which results in a stronger version of the entropic uncertainty relation of Białynicki-Birula and Mycielski \cite{BialynickiBirula1975}, see \cite{Hertz_2017} or \cite{Hertz_2019} for a review. However, extending this entropic uncertainty measure to states with negative Wigner functions appears precluded because it inherently comes with logarithms of negative numbers.

In the present work, we start by exploring a few possible ways of adapting the definition of Wigner entropy to Wigner-negative states, \textit{i.e.}, states such that $\exists \alpha\in\mathbb{C} : W(\alpha)<0$. Unfortunately, no good alternative real-valued functional  emerges that obeys all desirable properties of an entropy measure. Then, we advocate that the same functional should be kept when considering Wigner-negative states by resorting to the analytic continuation of the Shannon entropy in the complex plane. As we shall see, this leads us to define a complex-valued Wigner entropy, whose real and imaginary parts both have a physical meaning (the real part has an interpretation close to that of the ordinary real-valued entropy, while the imaginary part is simply proportional to the negative volume of the Wigner function). Interestingly, both the real and imaginary parts of the complex Wigner entropy  remain invariant under Gaussian unitaries. Finally, as an additional argument supporting this continuation in the complex plane, we consider the complex extension of de~Bruijn's identity, which leads us to define a complex-valued Fisher information associated with any Wigner function.

This paper is organized as follows. In Sec. \ref{sec:Wigner-entropy}
we recall the basics of Wigner functions and summarize the main facts about the (real-valued) Wigner entropy of Wigner-positive states \cite{PhysRevA.104.042211}. In Sec. \ref{sec:alternative} we consider some possible ways of extending this definition to Wigner-negative states while keeping a real-valued entropy. We list the desirable properties of such an extended 
entropy and discuss the results of our exploration, suggesting that a real-valued entropy functional cannot be found. In Sec. \ref{sec:Complex-Wigner}
we then make the step to extend the Wigner entropy to the complex plane. We prove some main properties of the complex Wigner entropy 
and discuss a few examples.
In Sec. \ref{sec:Complex-Fisher} we move to the complex-valued Fisher information and prove that its real (resp. imaginary) part governs the time-evolution of the real (resp. imaginary) part of the Wigner entropy under Gaussian additive noise. Finally, we discuss the perspectives of our work in Sec. 
\ref{sec:Conclusion}.

\section{Wigner entropy}
\label{sec:Wigner-entropy}

\subsection{Wigner function}

The Wigner function associated with a density operator $\hat \rho$ is defined as 
\begin{equation}\label{eq:wignerdef1}
W(x,p) = {1\over 2\pi} \int \langle x- {y\over 2} | \hat \rho | x + {y\over 2} \rangle \, e^{ipy} \, \dd y,
\end{equation}
where we have set $\hbar=1$. Note that we use from now on the $(x,p)$ canonical coordinates instead of the complex variable $\alpha=(x+ip)/\sqrt{2}$. More generally, the Weyl transform $A (x,p)$ of any linear operator $\hat A$ is defined as
\begin{equation}
A(x,p) = {1\over 2\pi} \int \langle x- {y\over 2} | \hat A | x + {y\over 2} \rangle \, e^{ipy} \, \dd y,
\end{equation}
implying that the Wigner function is simply the Weyl transform of the density operator $\hat \rho$. It is straightforward to check that
\begin{align}
\hat A = \hat A^\dagger \quad &\Leftrightarrow\quad 
A(x,p) = A^*(x,p) ,  \nonumber  \\
\mathrm{Tr} \hat A \; &=\; \int  A(x,p) \, \mathrm dx \, \mathrm dp ,
\end{align}
which immediately implies that the Wigner function $W$ of any physical state $\hat\rho$ (with $\hat \rho = \hat \rho^\dagger$ and $\mathrm{Tr} \hat \rho = 1$) is real and normalized to one. Furthermore, note that the Weyl transforms of operators $f(\hat x)$ and $g(\hat p)$ are $f(x)/2\pi$ and $g(p)/2\pi$, respectively, where $f$ and $g$ are arbitrary functions.

For any two linear operators $\hat A_1$ and $\hat A_2$, the so-called overlap formula reads
\begin{equation} \label{eq:overlapformula}
\tr {\hat A_1 \,\hat A_2}
 =  2 \pi \int A_1(x,p) \, A_2(x,p) \, \dd x \, \dd p,
\end{equation}
which implies, for example, that
\begin{eqnarray} 
& & \tr{\hat \rho}
 =  \int W(x,p) \, \dd x \, \dd p , \nonumber \\
& & \tr{\hat \rho^2}
 =  2 \pi \int W(x,p)^2 \, \dd x \, \dd p ,
\end{eqnarray}
since the Weyl transform of the identity operator $\hat \openone$ is simply the constant function $1/2\pi$. Note that  
\begin{eqnarray} 
& & \tr{\hat \rho^s}
 \ne c\,  \int W(x,p)^s \, \dd x \, \dd p , \qquad \forall s\ne 1 \text{~or~} 2,
\end{eqnarray}
for any constant $c$. Furthermore, using the overlap formula and the fact that the Weyl transform of a projector $\ket{x_0}\bra{x_0}$ in position space  is $\delta(x-x_0)/2\pi$, and, similarly, a projector $\ket{p_0}\bra{p_0}$ in momentum space is $\delta(p-p_0)/2\pi$, we get
\begin{eqnarray} && \rho_x(x_0) \coloneqq \langle x_0 | \hat \rho | x_0 \rangle 
 =  \int W(x_0,p) \, \dd p ,\nonumber \\
&& \rho_p(p_0) \coloneqq \langle p_0 | \hat \rho | p_0 \rangle 
 =  \int W(x,p_0) \, \dd x ,  
\end{eqnarray}
confirming that the marginals of $W(x,p)$ yield the probability densities of measuring the coordinates $x$ or $p$, denoted respectively as $\rho_x$ and $\rho_p$. To be complete, note also that the Weyl transform of the particle-number parity operator $(-1)^{\hat{n}}$ is $\delta(x)\delta(p)/2$, so that the overlap formula yields the well-known formula
\begin{equation}
W(0,0) = \frac{1}{\pi}\, \tr{\hat\rho \,(-1)^{\hat{n}}} , 
\label{eq:Wigner_origin}\end{equation}
connecting the expectation value of the particle-number parity with the Wigner function at the origin. Denoting the displacement operator as $\hat D_{x',p'}$, the Weyl transform of $\hat D_{x',p'}(-1)^{\hat{n}}\hat D_{x',p'}^\dagger$ becomes $\delta(x-x')\delta(p-p')/2$, so we may extend Eq.~\eqref{eq:Wigner_origin} into 
\begin{equation}
W(x',p') = \frac{1}{\pi}\, \tr{\hat D^\dagger_{x',p'}\, \hat\rho \, \hat D_{x',p'}\,(-1)^{\hat{n}}}, \end{equation}
which implies that the Wigner function at $(x',p')$ is proportional to the average particle-number parity of the   displaced state by $(-x',-p')$. Hence, we also note  that $|W(x,p)|\le 1/\pi$, $\forall x,p$. Finally, by taking the Weyl transform of the identity
\begin{align}
W(x,p)=\int W(x',p') \,\delta(x-x')\delta(p-p') \, \dd x'\, \dd p' ,
\end{align}
we obtain
\begin{align}
\hat \rho= 2 \int W(x',p') \, \hat D_{x',p'}\, (-1)^{\hat{n}}\, \hat D_{x',p'}^\dagger\, \dd x'\, \dd p' ,
\end{align}
which is the inverse Weyl transform of $W(x',p')$.

\subsection{Wigner entropy of Wigner-positive states}
\label{sec:Wigner-entropy-positive-states}

In Shannon information theory, the differential entropy of a pair of random variables $(x,p)$ is defined as \cite{CoverThomas}
\begin{equation}  \label{eq:def:h(x,p)}
h(x,p)=-\int f(x,p)\ln f(x,p) \, \dd x \, \dd p
\end{equation}
where $f(x,p)$ is the joint probability density of $x$ and $p$. Roughly speaking, it measures how spread out is the distribution $f(x,p)$. Note that $h(x,p)$ can be positive but also negative. It satisfies the following scaling property: if the two random coordinates transform as
\begin{equation}
\begin{pmatrix}x'\\p'\end{pmatrix}=\mathcal{S}\begin{pmatrix}x\\p\end{pmatrix},
\end{equation}
where $\mathcal{S}$ is a $2\times 2$ matrix, then the differential entropy is shifted by an additive constant, namely \cite{CoverThomas}
\begin{equation}
\label{eq:scaling}
h(x',p')=h(x,p)+\ln|\!\det(\mathcal{S})|.
\end{equation}


As long as a state $\hat \rho$ admits a non-negative Wigner function, that is $W(x,p)\ge 0, \forall x,p$, the latter behaves as a genuine probability distribution (it is non-negative and normalized to one) and we may compute its Shannon differential entropy. Accordingly,  the Wigner entropy of a Wigner-positive state can be defined as \cite{PhysRevA.104.042211}\footnote{This entropy was originally defined in Ref. \cite{Hertz_2017} in an attempt to find stronger entropic uncertainty relations, but it was deemed problematic in view of its extension to Wigner-negative states.}
\begin{equation}  \label{eq:def:h}
h(W)=-\int W(x,p)\ln W(x,p) \, \dd x \, \dd p .
\end{equation}
Note that we write it as a functional of $W$ here instead of $h(x,p)$, but it measures the joint uncertainty of the ($x$,$p$) pair -- or actually any pair of canonically conjugate coordinates --  in phase space as long as the system admits a non-negative Wigner function $W$. We immediately see from Eq. 
\eqref{eq:scaling} that the Wigner entropy is invariant under symplectic transformations (\textit{i.e.}, any phase shift or squeezing) since $\det(\mathcal{S})=1$ in such a case \cite{RevModPhys.84.621}. It is also trivially invariant under a displacement in phase space since $h(x+c)=h(x)$ for any constant $c\in\mathbb{R}$ \cite{CoverThomas}. Hence, $h(W)$ is invariant under all Gaussian unitaries. Note also that it takes the value $\ln \pi +1$ for the vacuum state as well as any pure Gaussian state in view of this invariance.

The Wigner entropy appears as a natural functional in order to measure the uncertainty of Wigner-positive states in phase space as it enjoys reasonable properties, most notably:
\begin{itemize}
	\item \textit{Symmetric functional} -- The Wigner entropy can be rewritten as
\begin{equation}
h(W)=\int \varphi\big(W(x,p)\big) \, \dd x \, \dd p ,
\label{eq:h(W)symmetric}
\end{equation}
where $\varphi(z)=-z \ln z$. Therefore, it is indeed a symmetric functional, which means that it is invariant under any area-preserving transformation~$\mathcal{M}$, that is:
$h(\mathcal{M}\left[W\right])=h(W)$, where $\mathcal{M}[W]$ is any transformation in phase space that keeps the level-function of $W$ unchanged \cite{VanHerstraeten2023continuous}. This property is the continuous counterpart to the fact that the discrete Shannon entropy is invariant under permutations of the probability vector. The set of area-preserving transformations $\mathcal{M}$ includes the class of Gaussian unitaries (symplectic transformations and displacements in phase space). 
 \item \textit{Concave functional} -- The Wigner entropy is a concave functional of the Wigner function, that is,
 \begin{equation}
h(\lambda_1 W_1+\lambda_2 W_2)\ge \lambda_1 \, h(W_1)+\lambda_2 \, h(W_2),
\end{equation}
where $\lambda_1,\lambda_2\ge 0$ and $\lambda_1+\lambda_2=1$.
This property is guaranteed by the fact that the function $\varphi(z)$ appearing in Eq. \eqref{eq:h(W)symmetric} is itself concave.
\item \textit{Lower bound on marginal entropies} -- The Wigner entropy is a lower bound on the sum of the entropies of the marginal distributions, namely, 
\begin{equation}
 h(W)\leq h(\rho_x)+h(\rho_p).
\label{eq:subbaditivity_wigner_entropy}
\end{equation}
This is a consequence of the subadditivity of Shannon's differential entropy \cite{CoverThomas}: the joint entropy can be written as $h(x,p)=h(x)+h(p)-I(x{\rm :}p)$, where the mutual information $I(x{\rm :}p) \ge 0$.
	\item \textit{Pure Gaussian extremal states}
	-- Although only  conjectured as of today \cite{PhysRevA.104.042211}, it is expected that the Wigner entropy of Wigner-positive states reaches its minimum over the set of pure Gaussian states, viewed as minimum-uncertainty states, that is
\begin{equation}
\label{eq:conjecture}
 \ln\pi +1 \stackrel{?}{\leq} h(W) .
\end{equation}
It must be recalled that the entropic uncertainty relation of Białynicki-Birula and Mycielski \cite{BialynickiBirula1975},
\begin{equation}
 \ln\pi +1 \leq h(\rho_x)+h(\rho_p),
 \label{eq:Bialynicki}
\end{equation}
gives a lower bound on the sum of marginal entropies which actually coincides with the Wigner entropy of pure Gaussian states. 
Provided it holds\footnote{Further progress on the proof of conjecture \eqref{eq:conjecture} has recently been reported in Ref. \cite{Dias2023-qr}, see also Ref. \cite{Zach-BS-states}.}, Eq.~\eqref{eq:conjecture} can thus be viewed as a tighter entropic uncertainty relation involving the Wigner entropy (limited to Wigner-positive states) instead of the sum of marginal entropies as in Eq. \eqref{eq:Bialynicki}.


		\end{itemize}


\section{Tentative  definitions for the entropy of Wigner-negative states}
\label{sec:alternative}

\subsection{Desirable properties}

We wish to construct a valid extension of $h(W)$ to Wigner-negative states that remains consistent with all (or most) of its properties. Let us denote the set of Wigner functions by $\mathcal{W}$, and the subset of non-negative Wigner functions by $\mathcal{W}_+$.
We seek a real-valued functional $\Phi:\mathcal{W}\mapsto\mathbb{R}$ with the following properties:
\begin{enumerate}
	\item $\forall W\in\mathcal{W}_+:\quad \Phi(W)=h(W)$
\item $\forall W\in\mathcal{W}:\quad\  \ \Phi(W)$ is a symmetric functional
\item $\forall W\in\mathcal{W}: \quad\  \  \Phi(W)$ is a concave functional
		\item $\forall W\in\mathcal{W}:\quad\  \  \Phi(W) \leq h(\rho_x)+h(\rho_p) $
	\item $\forall W\in\mathcal{W}:\quad\  \  \ln\pi+1 \leq \Phi(W)$.
\end{enumerate}

\subsection{Candidate real-valued extensions of $h(W)$}

We consider three possible real-valued extensions of $h(W)$ to Wigner-negative states which do satisfy properties $1$ and $2$ (\textit{i.e.}, symmetric functionals that reduce to the usual Wigner entropy for Wigner-positive states).
We first define the \textit{real} Wigner entropy as
\begin{equation}
	h_{\mathrm{r}}(W)
	=
	-\int W(x,p)
	\ln\abs{W(x,p)} \, \dd x \, \dd p .
\end{equation}
It is not an entropy \textit{stricto sensu} as it departs from the usual ``$-z\ln z$'' form, but we use this name because it will appear in Sec. \ref{sec:Complex-Wigner} that it  coincides with the real part of the complex-valued Wigner entropy.
We then define the \textit{absolute} Wigner entropy as the entropy of the absolute value of the Wigner function, namely
\begin{equation}
	h_{\mathrm{a}}(W)
	=
	-\int\abs{W(x,p)}
	\ln\abs{W(x,p)}\, \dd x \, \dd p .
\end{equation}
Note that $\abs{W(x,p)}$ is not normalized for Wigner-negative states, so that $h_{\mathrm{a}}(W)$ is also not a true entropy. Finally, we define the \textit{positive} Wigner entropy as the entropy of the positive part of the Wigner function, namely
\begin{equation}
	h_{\mathrm{+}}(W)
	=
	-\int\limits_{W(x,p)\geq 0} W(x,p)
	\ln W(x,p) \, \dd x \,\dd p.
\end{equation}
Note again that $W(x,p)$ is not normalized over the domain $W(x,p)\geq 0$ in the case of Wigner-negative states, so that $h_{\mathrm{+}}(W)$ is not a genuine entropy

\begin{figure}[t!]
\centering
\includegraphics[width=0.95\linewidth]{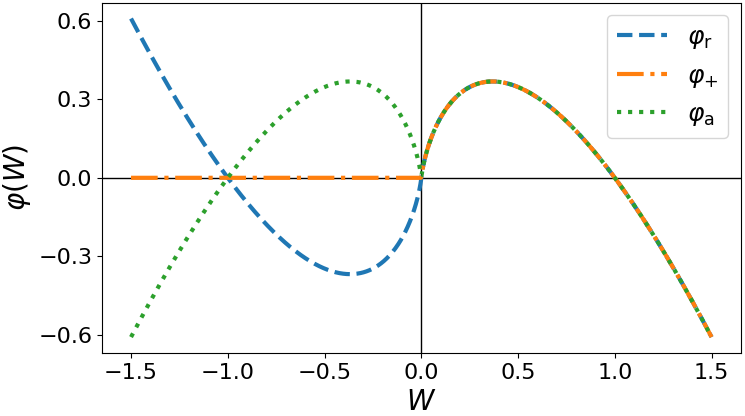}
\caption{
The three plotted functions $\varphi_r(W)=-W\ln\abs{W}$, $\varphi_a(W)=-\abs{W}\ln\abs{W}$, and $\varphi_+(W)=-\Theta(W)\, W\ln W$ are concave in $\mathbb{R}^+$ but not concave in $\mathbb{R}$. Of course, we have $\varphi_r(W)=\varphi_+(W)=\varphi_a(W)=-W\ln W$ provided $W\geq 0$.
These functions are used to build the symmetric functionals $h_{\mathrm{r}}(W)=\int\varphi_r(W(x,p))\, \dd x\, \dd p$, $h_{\mathrm{a}}(W)=\int\varphi_a(W(x,p))\,\dd x\, \dd p$ and $h_+(W)=\int\varphi_+(W(x,p))\,\dd x\,\dd p$, which are concave over the set $\mathcal{W}_+$ but not concave over the set $\mathcal{W}$. Note that the functions $\varphi(W)$ must only be defined on a narrow domain since $|W|\leq1/\pi$ but this does not change the problem.}
\label{Fig:phi_abs_pos_re}
\end{figure}


All three functionals $h_{\mathrm{r}}$, $h_{\mathrm{a}}$, and
$h_+$ are symmetric as they are constructed from the integration of some function $\varphi(W)$ over the entire phase space (see Fig. \ref{Fig:phi_abs_pos_re}), so they are invariant under Gaussian unitaries. Moreover, it is obvious that the three functionals coincide with the Wigner entropy when $W$ is non-negative, that is,  $h_{\mathrm{r}}(W)=h_{\mathrm{a}}(W)=h_+(W)=h(W)$, for all $W\in\mathcal{W}_+$. However, it appears that none of these three functionals is concave over the full set $\mathcal{W}$. This can be understood by looking at the corresponding functions $\varphi_r$, $\varphi_a$, and $\varphi_+$ in Fig. \ref{Fig:phi_abs_pos_re}. 
These three functions coincide with $\varphi(z)=-z \ln z$ when $z\ge 0$, so they are concave over the restricted domain $z\in\mathbb{R}^+$, but they are not concave anymore in the whole domain $z\in\mathbb{R}$. Hence, 
the functionals $h_{\mathrm{r}}$, $h_{\mathrm{a}}$, and $h_+$ are concave for $W\in\mathcal{W}_+$, but not concave for $W\in\mathcal{W}$. These functionals are, of course, not unique -- other functions $\varphi(z)$ could be built -- but it appears that the construction of a symmetric concave functional over $\mathcal{W}$ is bound to fail because it is impossible to build a concave function that extends the function $\varphi(z)=-z \ln z$ to negative $z$ (this is because the derivative of $\varphi(z)$ is infinite at $z=0^+$).

The construction of a real-valued extension of $h(W)$ satisfying properties 1, 2, and 3 thus seems hopeless, which motives the introduction of a complex-valued extension of $h(W)$ in Sec. \ref{sec:Complex-Wigner}. Before doing so, we examine whether some of the above three real-valued extensions of $h(W)$ may nevertheless satisfy properties 4 or 5. Note first that these three functionals can be ordered as
\begin{equation}
	h_{\mathrm{r}}(W)\leq h_+(W)\leq h_{\mathrm{a}}(W).
\label{eq:inequality_hr_ha}
\end{equation}
This can easily be proven by defining the \textit{negative} Wigner entropy as 
\begin{equation}
	h_{\mathrm{-}}(W)
	=	\int\limits_{W(x,p)\leq 0} W(x,p)
	\ln\abs{W(x,p)} \, \dd x \, \dd p  ,
\end{equation}
and noting that 
$h_{\mathrm{-}}(W)\ge 0$ since $\abs{W(x,p)}\leq 1/\pi$, so that $\ln\abs{W(x,p)}< 0$.
It is easily seen that the following relations hold:
\begin{align}
h_{\mathrm{r}}(W)
	=
	h_+(W)-h_-(W),
	\\
h_{\mathrm{a}}(W)
	=
h_+(W)+h_-(W),
\end{align}
which implies Eq. \eqref{eq:inequality_hr_ha} since $h_{\mathrm{-}}(W)\ge 0$. We have investigated whether these three functionals may be a lower bound to the sum of marginal entropies (property~4) or whether they may themselves be lower bounded by $\ln\pi+1$ (property~5). Since the results of this exploration are not fully satisfactory, we report them in Appendix \ref{sec:real-valued}. From numerical evidence, it seems that $h_+(W)$ is lower bounded by $\ln \pi+1$, hence the same is true for $h_{\mathrm{a}}(W)$, but we lack a proof. We can also lower bound the sum of marginal entropies $h(\rho_x)+h(\rho_p)$ by a quantity that depends on $h_{\mathrm{+}}(W)$, but the result is not tight.


\section{Complex-valued entropy of a~Wigner function}
\label{sec:Complex-Wigner}

\subsection{Definition}
The definition \eqref{eq:def:h} based on the differential Shannon entropy is only valid provided $W(x,p)\in\mathcal{W}_+$ because it requires taking the logarithm of $W(x,p)$ and the logarithm function in the real domain is not defined for negative values.
In order to extend the Wigner entropy to all states, including those with negative Wigner functions, we turn to a complex extension of the Wigner entropy by using the complex logarithm function
\begin{equation}
\text{Ln} \,(z)=\ln |z| +i \, \arg(z), \qquad z\in\mathbb{C}.
\label{eq:def_complex_logarithm}
\end{equation}
We must recall here that Ln($z$) is a multivalued function since the argument is defined modulo $2\pi$. Here, we conventionally use the principal value of the logarithm, meaning that its imaginary part lies in $]-\pi,\pi ]$. We will also use a lazy notation and write $\ln (z)$ instead of  $\text{Ln} \,(z)$ for the complex logarithm of $z$.

Thus, we define the complex Wigner entropy as 
\begin{eqnarray}
h_{\mathrm{c}}(W)&=&-\int  \, W(x,p) \, \ln W(x,p) \, \mathrm dx \, \mathrm dp \nonumber\\
&=& h_{\mathrm{r}}(W) + i \, h_{\mathrm{i}}(W)
\label{eq:complex-Wigner-entropy}
\end{eqnarray}
where the real and imaginary parts of $h_{\mathrm{c}}(W)$ are respectively given by
\begin{eqnarray}
h_{\mathrm{r}}(W) &=& -\int W(x,p) \, \ln|W(x,p)| \, \mathrm dx \, \mathrm dp  \\
h_{\mathrm{i}}(W)&=& -\int W(x,p) \, \arg W(x,p) \, \mathrm dx \, \mathrm dp \nonumber \\
&=& \pi\int\frac{|W(x,p)|-W(x,p)}{2}\, \mathrm dx \, \mathrm dp.
\label{realandim}
\end{eqnarray}
Note that $\arg W(x,p)$ is equal to $0$ if the Wigner function is positive and is equal to $\pi$ if it is negative, hence the second definition of $h_{\mathrm{i}}(W)$ [Eq.~(\ref{realandim})]. Thus, by convention, we define $h_{\mathrm{i}}(W)$ as non-negative.

The imaginary part of the complex Wigner entropy happens to be simply proportional to the volume of the negative part of the Wigner function, namely
\begin{eqnarray}
h_{\mathrm{i}}(W) = \pi \, \text{Vol}_{-}(W)
\end{eqnarray}
with
\begin{equation}
	\Vol_{-}(W)
	=-\int\limits_{W(x,p) < 0} W(x,p) \, \mathrm dx \, \mathrm dp.
\end{equation}
The negativity of the Wigner function has long been known to witness the non-classicality of a state and, along this line, the negative volume $\Vol_{-}(W)$ has even been used to build a measure of non-classicality \cite{Anatole_Kenfack_2004} (see also \cite{PhysRevA.98.052350,WitnWigNeg} for a resource-theoretical approach to Wigner negativity). It is remarkable that the negative volume appears naturally here simply from the expression of the complex Wigner entropy.
Note that $h_{\mathrm{i}}(W)$ is unbounded as it is known that one can build Wigner functions with an arbitrarily large negative volume \cite{10.1063/1.3486068}.

\subsection{Properties of the complex Wigner entropy}

Let us examine some properties of the real and imaginary parts of the complex Wigner entropy $h_{\mathrm{c}}(W)$. Note first that $h_{\mathrm{c}}(W)$ is defined only if the integral in Eq.~\eqref{eq:complex-Wigner-entropy} exists.
Since any Wigner function tends to zero at the limit of infinite distance from the origin in phase space, we have the complex logarithm of $W(x,p)\to 0$. However, using the standard continuity argument, we see that this does not cause a problem since $\lim_{W\to 0}W \ln |W| =0$ and $\lim_{W\to 0}W \arg W = 0$.

\begin{prop}
	Both the real and imaginary parts of the complex Wigner entropy are invariant under symplectic transformations  (Gaussian unitaries).
  \label{invariance-real-entropy}
\end{prop}
\begin{proof}
	Consider the symplectic transformation such that the vector of canonically conjugate variables $\mathbf{X}=(x,p)^T$ transforms into $\mathbf{X'}=(x',p')^T$ according to $\mathbf{X'}=\mathcal{S}\mathbf{X}+\mathbf{d}$, with $\mathcal{S}$ being a symplectic\footnote{A symplectic matrix $\mathcal{S}$ is a real matrix that preserves the so-called symplectic form $\Omega=\begin{psmallmatrix}
	    0 & 1 \\
	     -1 & 0
	\end{psmallmatrix}$, that is, $\mathcal{S}\Omega\mathcal{S}^T=\Omega$. Physically, this corresponds to rotations and squeezing in phase space.} matrix and $\mathbf{d}$ being a (real) displacement vector. Following this transformation, the Wigner function can thus be expressed as $W'(x',p')= |\det(\mathcal{S})|^{-1}W(x,p)$, where $(x,p)$ is expressed in terms of $(x',p')$ by inverting the transformation. Therefore, using the change of variable $\dd x'\, \dd p' = |\det(\mathcal{S})|\, \dd x\, \dd p$,
	\begin{eqnarray}
	h_{\mathrm{r}}(W')&=&-\int W'(x',p')\ln|W'(x',p')|\, \dd x'\, \dd p'\nonumber\\
	&=&-\int W(x,p)\ln \left|\frac{W(x,p)}{\det( \mathcal{S})}\right| \dd x \, \dd p\nonumber\\
	&=&h_{\mathrm{r}}(W)+ \ln |\det(\mathcal{S})| \nonumber\\
	&=& h_{\mathrm{r}}(W)
 \label{eq:proof-invariance-real-entropy}
	\end{eqnarray}
	and
	\begin{eqnarray}
	h_{\mathrm{i}}(W')&=&\pi\int\frac{|W'(x',p')|-W'(x',p')}{2} \, \dd x'\, \dd p'\nonumber\\
	&=&\pi\int\frac{|W(x,p)|-W(x,p)}{2 |\det( \mathcal{S})|}\, \dd x'\, \dd p'\nonumber\\
 &=&\pi\int\frac{|W(x,p)|-W(x,p)}{2 }\, \dd x\, \dd p\nonumber\\
	&=&h_{\mathrm{i}}(W).
 \label{eq:proof-invariance-imaginary-entropy}
	\end{eqnarray}
 \end{proof}
 
 Note that the proof of Eq.~\eqref{eq:proof-invariance-real-entropy} looks similar to the corresponding proof for Wigner-positive states in Sec. \ref{sec:Wigner-entropy-positive-states} as we 
 have used $\det(\mathcal{S})=1$
 since $\mathcal{S}$ is symplectic. Yet, the proof of Eq.~\eqref{eq:proof-invariance-real-entropy} exploits the normalization of $W(x,p)$ in $\cal W$ (encompassing Wigner-negative states).

Together, Eqs. 
 \eqref{eq:proof-invariance-real-entropy} and  \eqref{eq:proof-invariance-imaginary-entropy} thus imply that the complex Wigner entropy $h_{\mathrm{c}}$ is indeed invariant under all Gaussian unitaries. Note that this property remains valid when extending the Wigner function to $n$ modes.

\begin{prop}
Both the real and imaginary parts of the complex Wigner entropy are symmetric functionals.
\end{prop}
\begin{proof}
The proof simply goes by noting that the real and imaginary part of $h_{\mathrm{c}}$ can be written as \begin{eqnarray}
h_{\mathrm{r}}(W)&=&\int \varphi_r\left(W(x,p)\right) \, \dd x \, \dd p , 
\label{eq:h_r-phi_r}\\
h_{\mathrm{i}}(W)&=&\int \varphi_i\left(W(x,p)\right) \, \dd x \, \dd p ,
\label{eq:h_i-phi_i}
\end{eqnarray}
with 
$\varphi_r(z)=-|z|\ln z$ and $\varphi_i(z)=\pi(|z|-z)/2$. 
The fact that $W$ can be negative does not make any difference. 
\end{proof} 

It can be shown that integrals such as those appearing in Eqs. \eqref{eq:h_r-phi_r} and \eqref{eq:h_i-phi_i}
only depend on $W(x,p)$ through the level function, hence they are equal for two Wigner functions that have the same level function \cite{VanHerstraeten2023continuous}. This property extends the invariance of $h_{\mathrm{c}}$ under symplectic transformations (which are area-preserving) to the invariance of $h_{\mathrm{c}}$ under all (even nonphysical) area-preserving transformations (those that keep the level function of the Wigner function unchanged, so that the initial and final Wigner functions are equivalent in the sense of continuous majorization theory, see \cite{VanHerstraeten2023continuous}).

\begin{prop}
	The real part of the complex Wigner entropy is additive when considering a product state.
\end{prop}
\begin{proof}Let  us consider a two-mode Wigner function defined as $W(x_1,p_1,x_2,p_2)=W_1(x_1,p_1)W_2(x_2,p_2)$. Then,
	\begin{eqnarray}
	  h_{\mathrm{r}}(W)
	&=&-\int W\ln |W|\,\dd x_1 \,\dd p_1\, \dd x_2\, \dd p_2  \nonumber\\ 
	&=&-\int W_1\ln|W_1| \,\dd x_1 \,\dd p_1\int W_2\, \dd x_2\, \dd p_2 \nonumber\\
	&&~-\int W_1\,\dd x_1 \,\dd p_1 \int W_2\ln|W_2|\, \dd x_2\, \dd p_2\nonumber\\
	&=&h_{\mathrm{r}}(W_1) + h_{\mathrm{r}}(W_2)
	\end{eqnarray}
	where we have omitted the arguments of Wigner functions for notational simplicity. 
\end{proof}

Again, the proof is similar to that for Wigner-positive states but here we have used the normalization of Wigner functions in $\cal W$, including negative Wigner functions.
Interestingly, numerics seems to indicate that the real part of the complex Wigner entropy is subadditive for a two-mode Wigner function $W(x_1,p_1,x_2,p_2)$ that is not in a product form, namely $h_{\mathrm{r}}(W)\leq h_{\mathrm{r}}(W_1)+h(W_2)$, but we have not found a proof.

\begin{prop}
	The imaginary part of the complex Wigner entropy is superadditive when considering a product state.
\end{prop}
\begin{proof}Using the identity 
\begin{eqnarray}
\lefteqn{ |XY|-XY = (|X|-X)Y+X(|Y|-Y) }\hspace{2cm} \nonumber \\
&+&(|X|-X)(|Y|-Y) 
\end{eqnarray}
and setting $W(x_1,p_1,x_2,p_2)=W_1(x_1,p_1)W_2(x_2,p_2)$, we obtain
	\begin{eqnarray}
	\lefteqn{h_{\mathrm{i}}(W) = \pi\int \frac{|W|-W}{2} \,\dd x_1 \,\dd p_1\, \dd x_2\, \dd p_2} \nonumber\\
	&=&\pi\int \frac{|W_1 W_2|-W_1 W_2}{2}\,\dd x_1 \,\dd p_1\, \dd x_2\, \dd p_2  \nonumber\\
 &=&\pi\int \frac{(|W_1|-W_1)}{2}\,\dd x_1 \,\dd p_1 \int W_2\, \dd x_2\, \dd p_2\nonumber\\
 &&~+\pi\int W_1\, \dd x_1\, \dd p_1 \int \frac{(|W_2|-W_2)}{2}\,\dd x_2 \,\dd p_2 + \Delta\nonumber\\
 &=&h_{\mathrm{i}}(W_1)+h_{\mathrm{i}}(W_2) + \Delta 
\label{eq:superadditivity}
\end{eqnarray}
where we have defined
\begin{eqnarray}
\Delta &=& \pi\int\frac{(|W_1|-W_1)(|W_2|-W_2)}{2}\,\dd x_1 \,\dd p_1\, \dd x_2\, \dd p_2
\nonumber \\
&=& \frac{2}{\pi}\, h_{\mathrm{i}}(W_1) \,h_{\mathrm{i}}(W_2)\nonumber\\  &\geq& 0 \end{eqnarray} 
\end{proof}
Thus, $h_{\mathrm{i}}(W)\geq h_{\mathrm{i}}(W_1)+h_{\mathrm{i}}(W_2)$. We see that $\Im(h_{\mathrm{c}})$ becomes additive in the special case where at least one of the two states is Wigner-positive, but otherwise it is superadditive because the product of the negative volumes is nonzero\footnote{Eq.~\eqref{eq:superadditivity} can be understood by re-expressing it in terms of Wigner negative volumes, namely $\text{Vol}_{-}(W)=\text{Vol}_{-}(W_1)+\text{Vol}_{-}(W_2)+2\, \text{Vol}_{-}(W_1)\, \text{Vol}_{-}(W_2)$. Indeed, by using normalization, \textit{i.e.}, $\text{Vol}_{+}(\cdot)-\text{Vol}_{-}(\cdot)=1$, Eq.~\eqref{eq:superadditivity} becomes equivalent to $\text{Vol}_{-}(W)=\text{Vol}_{+}(W_1)\,\text{Vol}_{-}(W_2)+ \text{Vol}_{-}(W_1)\,\text{Vol}_{+}(W_2)$. Clearly, $W<0$ if either $W_1<0$ or $W_2<0$ but not both. Note that superadditivity is consistent with the interpretation of $\text{Vol}_{-}(W)$ as a resource monotone.}. Note that the superadditivity of $\Im(h_{\mathrm{c}})$ is linked to our convention  $\arg W(x,p)=\pi$ for $W(x,p)<0$, but it should be realized that $\Im(h_{\mathrm{c}})$  becomes subadditive if we choose $\arg W(x,p)=-\pi$. From numerics, we also observe that the imaginary part of the complex entropy~$h_{\mathrm{i}}$ apparently remains superadditive for arbitrary two-mode Wigner functions $W(x_1,p_1,x_2,p_2)$, but we have not proven this.

\begin{prop}
	The imaginary part of the complex Wigner entropy is  a convex function of the Wigner function\footnote{This property can also be immediately deduced from the convexity of the negative volume, see Appendix C of \cite{PhysRevA.98.052350}. 
 }. 
\end{prop}
\begin{proof}Let  us consider the following mixture $W(x,p)=\lambda W_1(x,p)+(1-\lambda)W_2(x,p)$, where $\lambda \in [0,1]$. Then,
	\begin{eqnarray}
	h_{\mathrm{i}}(W)
	&=&\pi\int \left(\frac{|\lambda W_1+(1-\lambda)W_2|}{2}\right.\nonumber\\
	&& \left. ~~~~~~~~ -\frac{\lambda W_1+(1-\lambda)W_2}{2}\right) \, \dd x \, \dd p\nonumber\\
	&\leq&\pi\int \left(\frac{\lambda| W_1|+(1-\lambda)|W_2|}{2}\right.\nonumber\\
	&&\left. ~~~~~~~ -\frac{\lambda W_1+(1-\lambda)W_2}{2}\right) \, \dd x \, \dd p\nonumber\\
	&=&\lambda \, h_{\mathrm{i}}(W_1)+(1-\lambda) \, h_{\mathrm{i}}(W_2)
	\end{eqnarray}
	where we have used the triangle inequality, namely $|X+Y|\le |X|+|Y|$.
\end{proof}

Again, convexity  originates here from the choice that  $\arg W(x,p)=\pi$ for $W(x,p)<0$. The opposite choice $\arg W(x,p)=-\pi$ would instead imply the concavity of $\Im(h_{\mathrm{c}})$. Unfortunately, we cannot say anything about the concavity of the real part of the Wigner entropy. Numerical tests show that it is neither concave nor convex. 

\subsection{Expressing the Wigner entropy in state space}

By using the overlap formula \eqref{eq:overlapformula}, the definition \eqref{eq:complex-Wigner-entropy} of the complex Wigner entropy in phase space may be reexpressed in state space as
\begin{eqnarray}
h_{\mathrm{c}}(W)
= \tr{ \hat \rho \, \hat\Theta } ,
\end{eqnarray}
where $\hat\Theta$ is defined as the operator whose Weyl transform is given by
\begin{eqnarray}
\mathcal W[\hat\Theta] = - \frac{1}{2\pi}\ln W(x,p) 
\end{eqnarray}
Thus, we may also rewrite 
$\hat\Theta=\mathcal S[\hat L] $,
where $\mathcal S[\hat L ]$ denotes the Weyl (or symmetric) ordered form of the operator $\hat L =  - \ln W(\hat x,\hat p)$, hence 
\begin{eqnarray}
h_{\mathrm{c}}(W) = \tr{ \hat \rho \,  \mathcal S[\hat L] } .
\label{eq:Wigner-entropy-state-space}
\end{eqnarray}

Unfortunately, Eq. \eqref{eq:Wigner-entropy-state-space} does not seem to provide a very convenient method to compute $h_{\mathrm{c}}$, except in some special cases. As an example, consider a thermal state with a mean photon number $\nu$. Its Wigner function is
\begin{equation}
W_\nu(x,p) = \frac{1}{\pi(1+2\nu)} \exp(-\frac{x^2+p^2}{1+2\nu}) .
\end{equation}
Here, the operator
\begin{equation}
\hat L = -\ln W_\nu(\hat x, \hat p) = \ln \pi(1+2\nu) + \frac{\hat x^2+\hat p^2}{1+2\nu}
\end{equation}
is already written in symmetric order, so we have
\begin{equation}
\mathcal S[\hat L] = \ln \pi(1+2\nu) + \frac{\hat \openone + 2 \hat n}{1+2\nu} .
\end{equation}
Since $\tr{\hat \rho \, \hat n}=\nu$, Eq.~\eqref{eq:Wigner-entropy-state-space} gives
\begin{eqnarray}
h_{\mathrm{c}}(W_\nu)
&=& \ln \pi +1 + \ln (1+2\nu),
\end{eqnarray}
which is indeed the Wigner entropy of a Gaussian thermal state.

\subsection{Examples and numerical exploration}

We have numerically computed the complex Wigner entropy of a variety of quantum states in order to get a grasp on its meaning. First of all, we know that $h_{\mathrm{c}}$ lies on the real axis as soon as the state is Wigner positive. Furthermore, provided that conjecture \eqref{eq:conjecture} is true, we know that $\Re h_{\mathrm{c}} \geq \ln \pi +1$ for all Wigner-positive states. This is visible in Fig. \ref{fig:random-points}, where we have plotted the complex Wigner entropy $h_{\mathrm{c}}$ of randomly generated states\footnote{
The generation of random states works as follows. Draw an eigenvalue vector $\bm{\lambda}$ from  a random $N\times N$ unitary $\mathbf{U}$, so that $\lambda_n=\abs{U_{0n}}^2$ [$\bm{\lambda}=(1,0,...)$ for pure states], then apply another random $N\times N$ unitary $\mathbf{V}$ so that $\rho=\mathbf{V}\mathrm{diag}(\bm{\lambda})\mathbf{V}^\dagger$ (here $\mathbf{U},\mathbf{V}$ are drawn according to the Haar measure).
This defines a state $\hat{\rho}=\sum_{mn}\rho_{mn}\ket{m}\bra{n}$.
More details about this procedure are given in Appendix C of Ref. \cite{VanHerstraeten2023continuous}.
Random states of Fig. \ref{fig:random-points} are generated for $N=2,...,14$.
}, in particular pure states (orange) and mixed states (blue). The red dotted line indicates $h_{\mathrm{r}}=\ln\pi+1$. We note that mixed states typically have a lower value of $\Im h_{\mathrm{c}} $, which is of course associated with a lower (or even zero) negative volume. The first Fock states $\ket{n}$ with $n$ ranging from 0 to 8 are represented by red stars in the complex entropy plane; we observe that both $\Re h_{\mathrm{c}}$ and $\Im h_{\mathrm{c}}$ increase with $n$. The vacuum state (as well as any pure Gaussian state) corresponds to the point $(\ln \pi +1,0)$.

\begin{figure}[t!]
\centering
\includegraphics[width=\linewidth]{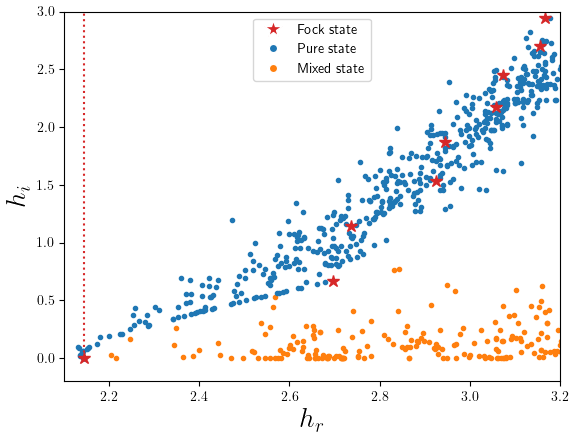}
\caption{Complex Wigner entropy $h_{\mathrm{c}} = h_{\mathrm{r}} + i \, h_{\mathrm{i}}$ of randomly generated states. Each blue (orange) point is associated with a random pure (mixed) state. Red stars represent Fock states $\ket{n}$ with $n=0,1,..., 8$. Note that $h_{\mathrm{r}}$ and $h_{\mathrm{i}}$ increase monotonically with $n$ for Fock states.
The dotted red line corresponds to $h_{\mathrm{r}}=\ln\pi+1\simeq 2.145$, which is a (conjectured) lower bound on the real part of the entropy of Wigner-positive states \cite{PhysRevA.104.042211}.
}
\label{fig:random-points}
\end{figure}

In Fig. \ref{fig:magnified-Fock-states}, we look more precisely at the complex entropy near this point and compare the superpositions and mixtures of the lowest Fock states. The situation is more subtle than anticipated as we see that for some pure state near the vacuum, $\Im h_{\mathrm{c}}$ becomes positive (which logically reflects the non-Gaussianity of the state as a consequence of Hudson's theorem) but $\Re h_{\mathrm{c}}$ goes below $\ln \pi +1$. This conjectured lower bound only applies along the real axis (\textit{i.e.}, for Wigner positive states), so it is not contradicted. Sadly, however, it seems difficult to determine the lowest possible value of $\Re h_{\mathrm{c}}$ for all states. Furthermore, we observe that the path connecting two neighboring Fock states is rather complicated. We see that the superpositions (solid blue curves) generally have much larger values of $\Im(h_{\mathrm{c}})$ than the corresponding mixtures (dashed orange curves), which follows from the comparison of their negative volumes. The value of $\Im(h_{\mathrm{c}})$ reaches precisely zero if the mixed state becomes Wigner positive (which happens if the weight of $\ket{0}$ is large enough). For example, we have $\Im(h_{\mathrm{c}})=0$ for balanced mixtures of $\ket{0}$ and $\ket{1}$, or balanced mixtures of $\ket{0}$ and $\ket{2}$, but never for any mixture of $\ket{1}$ and $\ket{2}$. When increasing the weight of $\ket{0}$ beyond this point, $h_{\mathrm{c}}$ follows the real axis down to the lower bound attained for $\ket{0}$.

\begin{figure}
    \centering
    \includegraphics[width=\linewidth]{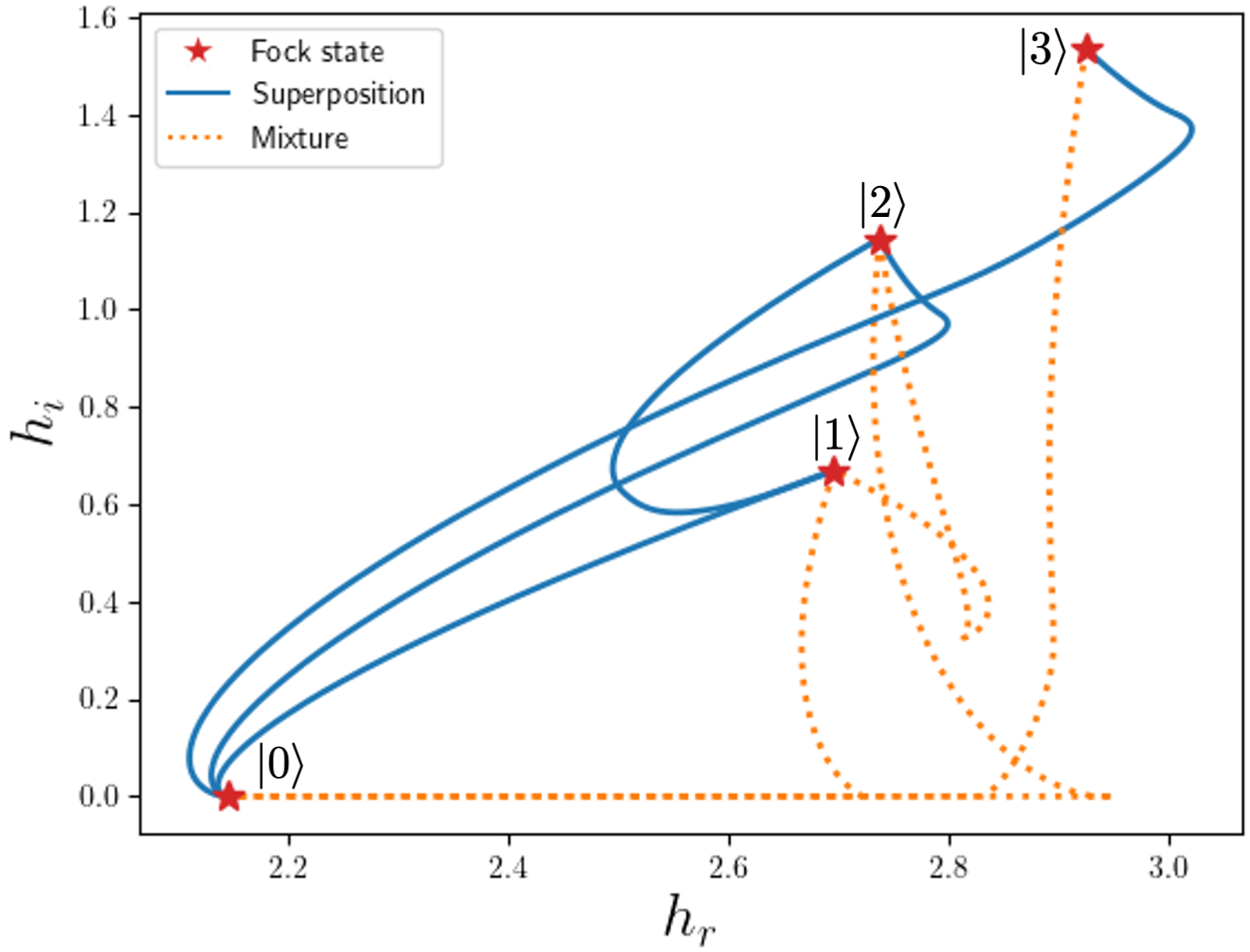}
    \caption{
    Plots of the complex Wigner entropy $h_{\mathrm{c}} = h_{\mathrm{r}} + i \, h_{\mathrm{i}}$ for superpositions vs. mixtures of Fock states (red stars stand for Fock states).
    The blue lines correspond to binary superpositions of Fock states, \textit{i.e.}, $\ket{\psi}=\sqrt{p}\ket{m}+e^{i\varphi}\sqrt{1-p}\ket{n}$ with $p\in[0,1]$ (note that the complex Wigner entropy of $\ket{\psi}$ is independent of $\varphi$). The dotted orange lines are the corresponding binary mixtures of Fock states, \textit{i.e.}, $\hat{\rho}=p\ket{m}\bra{m}+(1-p)\ket{n}\bra{n}$.
    The following binary superpositions and mixtures are plotted: $(m,n)=(0,1)$, $(0,2)$, $(0,3)$, and $(1,2)$.
    }
    \label{fig:magnified-Fock-states}
\end{figure}

Let us now perform the analytical calculation of $h_{\mathrm{c}}$ for another example, namely cat states. With the appropriate displacement and rotation (hence, some Gaussian unitary, which do not affect $h_{\mathrm{c}}$), any cat state can be brought to the following canonical expression:
\begin{align}
    \ket{\psi}_\alpha
    =
    \frac{1}{\sqrt{C}}
    \left(
    \sqrt{m}\ket{\alpha}+e^{i\varphi}\sqrt{1-m}\ket{-\alpha}
    \right)
\end{align}
where $\alpha\in\mathbb{R}_+$, $\varphi\in[0,2\pi)$, $m\in(0,1)$, and the normalization constant is  $C=1+2\sqrt{m(1-m)}\exp(-2\alpha^2)\cos\varphi$. The density operator is thus
\begin{align}
    &\ket{\psi}_\alpha\!\bra{\psi}
    =
    \frac{1}{C}\Big(
    m\ket{\alpha}\bra{\alpha}+(1-m)\ket{-\alpha}\bra{-\alpha}
    \nonumber
    \\
    &~~+2\sqrt{m(1-m)} (e^{i\varphi}\ket{-\alpha}\bra{\alpha} + e^{-i\varphi}\ket{\alpha}\bra{-\alpha})
    \Big)
\end{align}
so the corresponding Wigner function is given by
\begin{align}
    W_\alpha(x,p)
    =&
    \frac{m}{C}\,
    W_0(x-\sqrt{2}\alpha,p)
    +
    \frac{1-m}{C}\,
    W_0(x+\sqrt{2}\alpha,p)
    \nonumber
    \\&+
    \frac{2\sqrt{m(1-m)}}{C} \,
    W_0(x,p)
    \cos(2\sqrt{2}\alpha p+\varphi)
\end{align}
where $W_0(x,p)=\exp(-x^2-p^2)/\pi$ is the Wigner function of the vacuum state. In Fig. \ref{fig:cat-states}, we plot the real and imaginary parts of the complex Wigner entropy of this cat state (taking $m=1/2$ and $\varphi=0$) and compare with their values for the corresponding incoherent mixture of the coherent states $\ket{\alpha}$ and
$\ket{-\alpha}$. Of course, $\Im (h_{\mathrm{c}})$ vanishes in the incoherent case, while it is non-zero for cat states (it is shown in Appendix \ref{sec:cat-states-high-alpha} that $\Im h_{\mathrm{c}}$ tends to 1 in the limit $\alpha\to\infty$). In contrast, the curves of $\Re (h_{\mathrm{c}})$ are more comparable for the cat state and its corresponding mixture (as shown in Appendix \ref{sec:cat-states-high-alpha}, the two curves tend to each other in the limit $\alpha\to\infty$, which is understandable as the two Wigner functions almost coincide up to a series of infinitely narrow fringes whose contributions cancel each other out). Of course, the cat state and mixture both coincide with the vacuum state for $\alpha=0$, hence their complex entropies coincide. It can also be verified that the difference between $h_{\mathrm{r}}$ for $\alpha=0$ and $h_{\mathrm{r}}$ for $\alpha\to\infty$ tends to $\ln 2$.

\begin{figure}[t]
    \centering
    \includegraphics[width=0.9\linewidth]{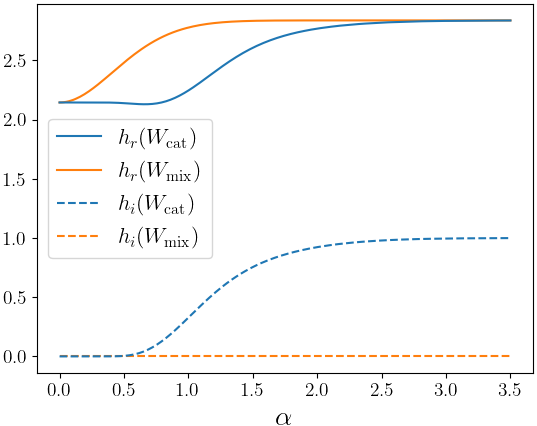}	
    \caption{Plots of the real and imaginary parts of the complex Wigner entropy $h_{\mathrm{r}}+i\, h_{\mathrm{i}}$ for $W_{\mathrm{cat}}$ and $W_{\mathrm{mix}}$, which are respectively the Wigner functions of an even cat state and of a balanced mixture of coherent states $\ket{\alpha}$ and $\ket{-\alpha}$. When $\alpha=0$, the two states coincide with the vacuum state. As $\alpha$ increases, the negative volume of the cat state increases. In the regime of $|\alpha|\gg 1$, $h_{\mathrm{r}}\left(W_{\mathrm{cat}}\right) \simeq h_{\mathrm{r}}(W_{\mathrm{mix}})$ and $h_{\mathrm{i}}\left(W_{\mathrm{cat}}\right)\simeq 1$.}
    \label{fig:cat-states}
\end{figure}


\section{Complex Fisher information and extended de Bruijn's identity}
\label{sec:Complex-Fisher}


\subsection{Classical overview}

Let us consider a vector $\mathbf{X}_0$ of classical random variables and a normally-distributed random vector $\mathbf{N}$ with covariance matrix $\mathbf{I}$.
From these two random vectors, we define the random vector $\mathbf{X}_t\vcentcolon=\mathbf{X}_0+\sqrt{t} \,\mathbf{N}$, which can be seen as a noisy version of $\mathbf{X}_0$ with the amount of noise growing in time $t$ (the variance increases linearly with $t$).
The probability density of $\mathbf{X}_t$ is the probability density of $\mathbf{X}_0$ convolved with a Gaussian distribution of covariance $t\mathbf{I}$, which describes a Gaussian diffusion process in time $t$ characterized by the diffusion equation
\begin{equation}
\frac{\dd}{\dd t} \, p_t(\mathbf{x}) = \frac{1}{2} \,  \Delta p_t  
\label{eq:diffusion-eq}
\end{equation}
where $\Delta=\bm{\nabla}^2$ is the Laplacian and $p_t(\mathbf{x})$ denotes
the probability density of $\mathbf{X}_t$.

De Bruijn's identity expresses the rate of increase of the Shannon entropy of the probability density $p_t(\mathbf{x})$ when it undergoes such a Gaussian diffusion. It is written as \cite{CoverThomas}
\begin{align}
\label{eq:deBruijn}
    \frac{\dd}{\dd t} \,
    h(\mathbf{X}_0+\sqrt{t}\mathbf{N})
    =
    \frac{1}{2} \, J(\mathbf{X}_0+\sqrt{t}\mathbf{N})
\end{align}
where $J(\mathbf{X})$ is the Fisher information of $\mathbf{X}$. The latter is defined based on the logarithmic gradient $\bm{\nabla}\ln p$, namely \begin{align}
J(\mathbf{X}) &=\int p(\mathbf{x}) \, \Vert \bm{\nabla}\ln p\Vert^2\dd\mathbf{x}, \label{eq:bad1}\\
&= \int  \bm{\nabla} p \cdot \bm{\nabla} \ln p \,\, \dd\mathbf{x}, \label{eq:good}\\
&= \int p^{-1}(\mathbf{x}) \, \Vert\bm{\nabla} p \Vert^2\dd\mathbf{x}, \label{eq:bad2}
\end{align}
where we have used $\bm{\nabla} \ln p =(\bm{\nabla} p ) / p$. Note that the above three definitions of the Fisher information are equivalent as long as $p(\mathbf{x})\ge 0$, $\forall \mathbf{x}$, but we shall see that only the second one is suitable when moving to the complex plane. De Bruijn's identity \eqref{eq:deBruijn} can be verified by explicitly calculating the time derivative of $h(\mathbf{X}_t)$ and using Eq.~\eqref{eq:diffusion-eq} as well as the definition of the Fisher information $J(\mathbf{X})$, see~\cite{CoverThomas}. Note also that $J(\mathbf{X})\ge 0$, expressing the fact that $\mathbf{X}_t$ can only spread over time. 

Note that de Bruijn's identity makes use of a special case of the definition of the Fisher information.
In general, the Fisher information is introduced for a family of distributions $\lbrace p_{\bm{\theta}}\rbrace$ which depend on a set of parameters $\bm{\theta}=(\theta_1,\theta_2,...)$.
The Fisher information is then defined as $J_{\bm{\theta}}(p_{\bm{\theta}})=\int p_{\bm{\theta}}(\mathbf{x}) \, \Vert\nabla_{\bm{\theta}}\ln p_{\bm{\theta}}(\mathbf{x})\Vert^2\dd\mathbf{x}$.
When the parameters $\bm{\theta}$ act as a translation over the distributions, so that $p_{\bm{\theta}}(\mathbf{x})=p(\mathbf{x}-\bm{\theta})$, the Fisher information simplifies to Eq.~\eqref{eq:bad1} and is independent of $\bm{\theta}$ \cite{CoverThomas}.
It is this particular non-parametric definition of the Fisher information that we will extend to the complex plane.

\subsection{Complex Fisher information}

Here, we will show that it is possible to derive a similar relation for an arbitrary Wigner function, even when it becomes negative.
In this case, both terms of de Bruijn's identity \eqref{eq:deBruijn} become complex-valued, so we need to introduce a complex-valued Fisher information too. Crucially, the diffusion equation \eqref{eq:diffusion-eq} holds regardless of the sign of $p(\mathbf{x})$, so it describes the evolution of any Wigner function under additive Gaussian noise (it is valid for positive as well as negative Wigner functions). Since we consider the Wigner function of a single mode, we deal with the 2-dimensional vector $\mathbf{x}=(x,p)^T$ of canonically conjugate coordinates. 
We define the complex Fisher information associated with any Wigner function $W(x,p)$ as the complex continuation of definition \eqref{eq:good}, that is,
\begin{align}
J_{\mathrm{c}}(W) &= \int  \bm{\nabla} W \cdot \bm{\nabla} \ln W \,\,  \dd x \, \dd p \nonumber\\
&= J_{\mathrm{r}}(W) + i \, J_{\mathrm{i}}(W),
\label{eq:def_complex_Fisher}
\end{align}
where $\ln W$ stands for the complex logarithm of $W$.
The real part of the complex-valued Fisher information is thus given by
\begin{align}
J_{\mathrm{r}}(W) &= \int  \bm{\nabla} W \cdot \bm{\nabla} \ln |W| \,\, \dd x \, \dd p .
\end{align}
One can write (notice the absence of absolute value on the right hand side)
\begin{equation}
\bm{\nabla} \ln |W| = \frac{1}{W} \, \bm{\nabla} W ,
\end{equation}
so we also recover the other two definitions
\begin{align}
J_{\mathrm{r}}(W) &= \int   W(x,p) \, \Vert \bm{\nabla} \ln |W| \, \Vert^2\,\, \dd x \, \dd p ,\nonumber \\
&= \int   W^{-1}(x,p) \, \Vert \bm{\nabla} W \Vert^2\,\, \dd x \, \dd p ,
\end{align}
which are the counterparts of Eqs. \eqref{eq:bad1} and \eqref{eq:bad2}.
These alternative definitions seem to suggest that $J_{\mathrm{r}}(W) \ge 0$ similarly as in the classical case, but we do not have a proof of this inequality (neither do we know if it holds).


From Eq.~\eqref{eq:def_complex_Fisher}, the imaginary part of the complex Fisher information is expressed as 
\begin{align}
J_{\mathrm{i}}(W) &= \int  \bm{\nabla} W \cdot \bm{\nabla} \arg W \,\, \dd x \, \dd p, \nonumber \\ &= \pi \int  \bm{\nabla} W \cdot \bm{\nabla} \mbox{\Large$\chi$} \,\, \dd x \, \dd p , 
\label{eq:imaginary_Fisher_info_initial}
\end{align}
where we have defined the Wigner-negativity indicator function
\begin{align}
    \mbox{\Large$\chi$}(x,p)
    \coloneqq \bm{1}_{(x,p)\in \cal D} =
    \begin{cases}
        0\qquad\text{if}\;W(x,p)\geq 0
        \\
        1\qquad\text{if}\;W(x,p)< 0
    \end{cases}
\end{align}
with $\cal D$ being the negative domain of $W$ in phase space (it may consist of several non-contiguous areas). Applying the nabla $\bm{\nabla}$ operator onto an indicator function gives the so-called surface delta-function, which vanishes everywhere except on the boundary $\partial D$ of $\cal D$, where it points in the (inward) normal  direction. It is well defined when integrated over phase space, so 
using the identity 
\begin{equation}
\bm{\nabla}(f \bm{G})=\bm{\nabla}f \cdot \bm{G} + f \,  \bm{\nabla} \bm{G}
\label{eq:nabla_identity}
\end{equation}
 and substituting $f$ with \mbox{\Large$\chi$} and $\bm{G}$ with $\bm{\nabla}W$, 
 we may rewrite Eq.~\eqref{eq:imaginary_Fisher_info_initial} as
 \begin{align}
J_{\mathrm{i}}(W) &= \pi \underbrace{\int 
\bm{\nabla}(\mbox{\Large$\chi$} \, \bm{\nabla}W) \,\, \dd x \, \dd p \,}_{\displaystyle \oint_{C_\infty} \!\!\! \mbox{\large$\chi$} \, \bm{\nabla}W \cdot \bm{n} \, \dd s  =0} - \, \pi \int 
\mbox{\Large$\chi$} \,  \Delta W \,\, \dd x \, \dd p , 
\end{align}
where the volume integral $\int \bm{\nabla}(\cdots)\, \dd x \, \dd p$ can be replaced by a contour integral $\oint (\cdots)\, \bm{n} \, \dd s$ on the contour $C_\infty$ at infinity, which vanishes since 
since 
$\mbox{\Large$\chi$}=0$ (and $\bm{\nabla}W$ is finite) on the contour $C_\infty$. Thus, the imaginary part of the Fisher information can be expressed as
\begin{align}
J_{\mathrm{i}}(W) &= - \pi \int 
\mbox{\Large$\chi$} \,  \Delta W \,\, \dd x \, \dd p , \nonumber \\
&= - \pi \int_{\cal D} \Delta W \,\, \dd x \, \dd p ,
\nonumber \\
&= - \pi \oint_{\partial D} 
\bm{\nabla} W \cdot \bm{n} \, \dd s . \label{eq:imaginary_Fisher_info_final}
\end{align}
which is a contour integral over the boundary $\partial D$ of the (possibly multiple) negative domain $\cal D$ in phase space. Here, $\bm{n}$ is a normal vector to the boundary (conventionally pointing outwards, hence the minus sign) and $\dd s$ denotes the infinitesimal element along the boundary. 
Note that $J_{\mathrm{i}}(W)\le 0$ since any Wigner function $W$ can only have (possibly multiple) negative domains in an overall positive infinite domain, thus $\bm{\nabla} W$ always points outwards along the boundary $\partial D$. 

As expected, the complex Fisher information enjoys  invariance under displacements and rotations in phase space (\textit{i.e.}, passive Gaussian unitaries), but not squeezing operations.
\begin{prop} The real and imaginary parts of $J_{\mathrm{c}}(W)$ are invariant under translations and orthogonal symplectic transformations in phase space. 
\end{prop}
\begin{proof}
   As in the proof of Property \ref{invariance-real-entropy}, we consider the symplectic transformation $\mathbf{X'}=\mathcal{S}\mathbf{X}+\mathbf{d}$ but we add the restriction that $\mathcal{S}$ is an orhogonal symplectic matrix, that is, $\mathcal{S}^{-1}=\mathcal{S}^T$. Then, 
   we use the change of basis formula for the nabla operator, namely $\bm{\nabla}=
   \mathcal{S}^T 
  \bm{\nabla'}$, which implies $\bm{\nabla'}=
   \mathcal{S}\, 
  \bm{\nabla}$ since $\mathcal{S}$ is orthogonal. Therefore,
\begin{eqnarray} J_{\mathrm{r}}(W')&=&\int(W'(x',p'))^{-1} \Vert \bm{\nabla'} W'(x',p') \Vert^2\,\, \dd x' \, \dd p'\nonumber\\
      &=&\int\frac{|\rm{det}\,\mathcal{S}|}{W(x,p)}\frac{\Vert \mathcal{S}\bm{\nabla}W(x,p)\Vert^2}{|\rm{det}\,\mathcal{S}|^2} |\rm{det}\,\mathcal{S}|\,\dd x\,\dd p\nonumber\\
      &=& J_{\mathrm{r}}(W).
\end{eqnarray}
For the imaginary part, we use the second definition in Eq. \eqref{eq:imaginary_Fisher_info_final} involving the Laplacian of $W$ and exploit the invariance of the latter operator under rigid motions, 
$\Delta'=\bm{\nabla'}^T\bm{\nabla'}=\bm{\nabla}^T \mathcal{S}^T \mathcal{S}\,\bm{\nabla}=\bm{\nabla}^T \bm{\nabla}=\Delta$, hence
\begin{eqnarray} 
      J_{\mathrm{i}}(W')&=&  - \pi \int_{\cal D'} \Delta' W'(x',p') \,\, \dd x' \, \dd p'\nonumber\\
      &=&  - \pi \int_{\cal D} \frac{\Delta W(x,p)}{|\rm{det}\,\mathcal{S}|} \, |\rm{det}\,\mathcal{S}|\,\, \dd x \, \dd p \nonumber\\
      &=&J_{\mathrm{i}}(W).
\end{eqnarray}
\end{proof}

For illustration, the value of the complex Fisher information has been plotted in Fig.~\ref{fig:complex_fishinfo} for randomly generated (pure and mixed) states as well as for the first few Fock states.

\begin{figure}[t]
    \centering
    \includegraphics[width=\linewidth]{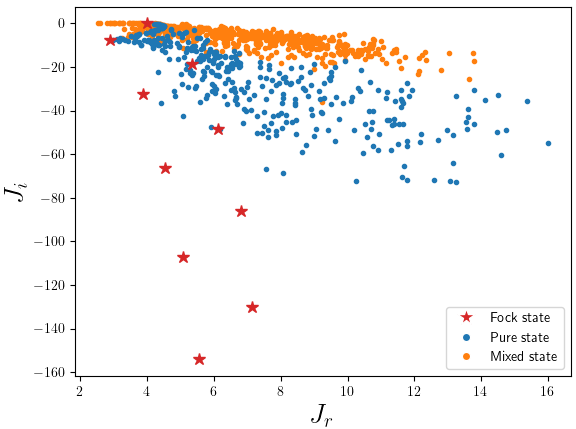}	
    \caption{Complex Fisher information $J_{\mathrm{c}} = J_{\mathrm{r}} + i \, J_{\mathrm{i}}$ of randomly generated states. Each blue (orange) point is associated with a random pure (mixed) state. Red stars represent Fock states $\ket{n}$ with $n=0,1,..., 9$. It appears that $J_{\mathrm{i}}$ decreases monotonically with $n$ for Fock states.}
    \label{fig:complex_fishinfo}
\end{figure}

\subsection{Complex de Bruijn's identity}

Let us now prove the complex version of de Bruijn's identity, which is expressed as 
\begin{equation}
    \frac{\dd}{\dd t}
    h_{\mathrm{c}}(W) = \frac{1}{2} J_{\mathrm{c}}(W) .
\label{eq:def_complex_deBruijn}
\end{equation}
\begin{proof}
We must calculate the time derivative of the complex Wigner entropy, namely
\begin{align}
    \lefteqn{ \frac{\dd}{\dd t}
    h_{\mathrm{c}}(W)
     }
    \nonumber
    \\
    &= -\frac{\dd}{\dd t}
    \int
    W
    \ln W \,
    \dd x \, \dd p \nonumber \\ &=
    -
    \int
    \left(
    \frac{\dd}{\dd t} W
    \right)
    \ln W \,
    \dd x \, \dd p
    -\underbrace{
    \int
    W \, 
    \frac{\dd}{\dd t}
    \ln W \,
    \dd x \, \dd p.
    }_{=0}
\label{eq:complex_deBruijn_1st}
    \end{align}
It is easy to show that both the real and imaginary parts of the second term of Eq. \eqref{eq:complex_deBruijn_1st} vanish. For the real part, 
    \begin{align}
 \int
    W \, 
    \frac{\dd}{\dd t} 
    \ln |W| \,
    \dd x \, \dd p &=
    \int \frac{\dd W}{\dd t}
     \,
    \dd x \, \dd p \nonumber \\ &= \frac{\dd }{\dd t}
     \underbrace{\int 
    W \,\dd x \, \dd p}_{=1} \nonumber\\ &=0 ,
    \end{align}
and for the imaginary part 
\begin{align}
       \lefteqn{ \int
    W \, 
    \frac{\dd}{\dd t} 
    \arg W \,
    \dd x \, \dd p } \hspace{1.5cm} \nonumber \\ &=\pi \int
    W \, 
    \frac{\dd}{\dd t} 
\left( \frac{1}{2} - \frac{|W|}{2W} \right)\,
    \dd x \, \dd p \nonumber \\ &=-\frac{\pi}{2} \int
\left(   \frac{\dd |W|}{\dd t} - \frac{|W|}{W} \, \frac{\dd W}{\dd t} \right) \,
    \dd x \, \dd p
\nonumber \\  
&=-\frac{\pi}{2} \int
\underbrace{\left(   \frac{\dd |W|}{\dd W} - \frac{|W|}{W}  \right)}_{=0} \, \frac{\dd W}{\dd t} \,
    \dd x \, \dd p
\nonumber \\ 
&=0.    \end{align}
This integral is zero simply because $W(x,p)$ vanishes along the boundary $\partial D$ of the negative domain $\cal D$.

Using the diffusion equation for the Wigner function (which is valid for all Wigner functions in $\cal W$), namely
\begin{equation}
\frac{\dd}{\dd t} W =  \frac{1}{2}\Delta W ,
\label{eq:diffusion_Wigner}
\end{equation}
we may rexpress Eq. \eqref{eq:complex_deBruijn_1st} as
    \begin{align}\lefteqn{ \frac{\dd}{\dd t}
    h_{\mathrm{c}}(W)
     }
    \nonumber
    \\
    &=
    -\frac{1}{2}
    \int
    \Delta W
    \ln W \,
    \dd x \, \dd p
    \nonumber
    \\
    &=
    -\frac{1}{2}
    \int
    \bm{\nabla}(\bm{\nabla}W
    \ln W) \, \dd x \, \dd p
    +\frac{1}{2}
    \int
    \bm{\nabla}W \!\cdot\! 
    \bm{\nabla}\ln W
    \, \dd x \, \dd p
    \nonumber
    \\
    &=
    -\frac{1}{2}
    \underbrace{
    \oint_{C_\infty}
    \ln W \, \bm{\nabla}W
    \cdot \bm{n} \, \dd s}_{=0}
    +\frac{1}{2}
    \underbrace{ \int
    \bm{\nabla}W
    \!\cdot\! \bm{\nabla}\ln W
    \, \dd x \, \dd p}_{J_{\mathrm{c}}(W)} ,
\end{align}
where we have used the identity \eqref{eq:nabla_identity}
substituting $f$ with $\ln W$ and $\bm{G}$ with $\bm{\nabla}W$.
The first term vanishes as it is a contour integral whose integrand $\ln W\, \bm{\nabla}W$ tends to zero at infinite distance from the origin in phase space. 
\end{proof}

It should be emphasized that this proof, although it looks partly similar to the corresponding one for probability distributions, applies to the complex plane (all logarithms should be understood as complex logarithms). 
Let us also mention that a proof for the real part of the complex de Bruijn's identity, Eq.~\eqref{eq:def_complex_deBruijn}, has been presented in Appendix A of Ref. \cite{pizzimenti2023complexvalued}.


Finally, we remark that the diffusion equation obeyed by the Wigner function, Eq.~\eqref{eq:diffusion_Wigner}, can be used in order to rewrite the second definition in
Eq. \eqref{eq:imaginary_Fisher_info_final} as
\begin{align}
J_{\mathrm{i}}(W) &= - 2\pi \int_{\cal D} \frac{\dd W}{\dd t} \,\, \dd x \, \dd p ,
\nonumber \\
&= - 2\pi \, \frac{\dd}{\dd t} \int_{\cal D} W \,\, \dd x \, \dd p , \nonumber \\
&= 2\pi \, \frac{\dd}{\dd t}\text{Vol}_{-}(W) , \nonumber \\
&=2 \, \frac{\dd}{\dd t} h_{\mathrm{i}}(W),
\end{align}
confirming the imaginary part of the complex de Bruijn's identity, Eq.~\eqref{eq:def_complex_deBruijn}.



\section{Discussion and perspectives}
\label{sec:Conclusion}

We have defined the complex Wigner entropy of an arbitrary state as the complex continuation of the Shannon differential entropy of its Wigner-function $W(x,p)$, namely
\begin{equation}
	h_{\mathrm{c}}(W)
	= -\int_{\mathbb{R}^2} W(x,p)
	\ln W(x,p)
	\, \dd x \, \dd p  ,
\end{equation}
where $\ln(\cdot)$ stands for the complex logarithm function. Following the same logic, the complex Fisher information has been defined as
\begin{equation}
J_{\mathrm{c}}(W) = \int_{\mathbb{R}^2}  \left(\frac{\dd W}{\dd x} \, \frac{\dd (\ln W)}{\dd x}  + \frac{\dd W}{\dd p} \,\frac{\dd (\ln W)}{\dd p} \right) \dd x \, \dd p
\end{equation}
These two definitions purposely coincide with the usual ones as soon as the state is Wigner positive, that is, $W(x,p)>0$, $\forall x,p$, since the Wigner function is then a genuine joint probability distribution. Thus, for such states, $\Im(h_{\mathrm{c}})=\Im(J_{\mathrm{c}})=0$, whereas   $\Re(h_{\mathrm{c}})$ and $\Re(J_{\mathrm{c}})$ inherit all properties of their counterparts in Shannon information theory. For example, the Wigner entropy is then a concave and subadditive functional of the state, and it obeys a Wigner entropy-power inequality \cite{PhysRevA.104.042211}.

In contrast, when the state is Wigner negative, that is, $\exists (x,p):W(x,p)<0$, its Wigner entropy admits a strictly positive imaginary part $\Im(h_{\mathrm{c}})$, which is proportional to the Wigner negative volume. Although the complex-valued entropy $h_{\mathrm{c}}$ loses some of its properties, such as concavity, it remains invariant when the state undergoes any Gaussian unitary (\textit{i.e.}, displacement, rotation, or squeezing in phase space). Both the real and imaginary parts of $h_{\mathrm{c}}$ can also be related to the real and imaginary parts of $J_{\mathrm{c}}$, respectively, when the state undergoes a Gaussian diffusion process (\textit{i.e.}, a convolution with a Gaussian noise whose variance grows linearly with time). As a consequence of a complex extension of de~Bruijn's identity, $\Re(h_{\mathrm{c}})$ evolves with a rate $\propto \Re(J_{\mathrm{c}})$ and, similarly, $\Im(h_{\mathrm{c}})$ evolves with a rate $\propto \Im(J_{\mathrm{c}})$. For a Wigner-positive state, $J_{\mathrm{c}}=\Re(J_{\mathrm{c}})\ge 0$ becomes the usual Fisher information, hence $h_{\mathrm{c}}=\Re(h_{\mathrm{c}})$ can only increase with time. This translates the fact that the Wigner function $W(x,p)$ expands with time. For a Wigner-negative state, $\Im(J_{\mathrm{c}})<0$, so that $\Im(h_{\mathrm{c}})$ can only decrease with time. This is consistent with the fact that the Wigner negative volume can only decrease along this Gaussian diffusion process, until it reaches zero. Once the state becomes Wigner-positive, $\Im(h_{\mathrm{c}})$  remains equal to zero along further time evolution. 

Note that both $\Im(h_{\mathrm{c}})$ and $\Im(J_{\mathrm{c}})$ have a straightforward interpretation involving the Wigner negative domain $\cal D$. As already mentioned, $\Im(h_{\mathrm{c}})$ is proportional to (minus) the volume integral of $W$ over $\cal D$, \textit{i.e.}, the negative volume, while $\Im(J_{\mathrm{c}})$ is proportional to (minus) the contour integral of the gradient of $W$ over the boundary $\partial D$ of the negative domain $\cal D$. Just as in Shannon information theory, the imaginary Fisher information can thus be linked to the derivative of the imaginary Wigner entropy. Note also that our analysis in this paper is restricted to Wigner functions associated with a single pair of canonically conjugate coordinates $(x,p)$, that is, a single bosonic mode, but our results should be easily generalizable to a multidimensional (multimode) case.

A possible application of the complex Wigner entropy may arise in the context of entropic uncertainty relations \cite{Hertz_2019}. For Wigner-positive states, we have 	\begin{equation}
h(\rho_x)+h(\rho_p)\stackrel{\mathrm{(a)}}{\geq}	\Re(h_{\mathrm{c}}(W))\stackrel{\mathrm{(b)}}{\geq}\ln\pi+1 ,
\label{eq:entropic-uncertainty-relation-conclusion}
	\end{equation}
where $\Re(h_{\mathrm{c}}(W))=h_{\mathrm{c}}(W)$ since it is real-valued.
More precisely, inequality (a) directly follows from the subadditivity of Shannon's differential entropy, while inequality (b) as conjectured in \cite{PhysRevA.104.042211} yields a tight entropic uncertainty relation \cite{Hertz_2017}. Of course, chaining inequalities (a) and (b) yields the entropic uncertainty relation due to Białynicki-Birula and Mycielski \cite{BialynickiBirula1975}. Clearly, Eq.~\eqref{eq:entropic-uncertainty-relation-conclusion} implies that there is a forbidden zone in the complex entropy plane along the real axis below $\ln\pi+1$. Moving to arbitrary (negative) Wigner functions, this strongly suggests that the complex Wigner entropy of a physical state is constrained to lie in some allowed area in the complex plane. The complex entropic uncertainty relation would then be expressed as $h_{\mathrm{c}} \in \mathcal{A}$, with $\mathcal{A}$ being the allowed area. This topic is worth further investigation.

Arguably, a weakness of our approach is that we miss a fully satisfactory understanding of the physical meaning behind the complex-valued entropy $h_{\mathrm{c}}$ (even its real part is not associated with a well understood property when the state is Wigner negative). We believe -- although we have not been able to do it -- that a satisfactory operational interpretation of the complex Wigner entropy could be obtained by extending the notion of typical volume to the complex plane. For a Wigner-positive state, Shannon information theory tells us that a sequence of $n$ independent instances of $W(x,p)$ populates with high probability a typical volume $\sim \exp(n\, h(W))$ in Euclidean space $\mathbb{R}^{2n}$. For a Wigner-negative state, it is tempting to keep the same expression and conclude that the typical volume becomes in general complex since the entropy is itself complex, translating the fact that for each instance of the sequence, $x$ and $p$ cannot be simultaneously defined (unlike two classical variables). Just as the entropy is related to the typical volume, the Fisher information is related to the surface area of the typical set, hence the complex Fisher information is another hint at the notion of complex typical volume.
If meaningful, this notion should of course be put on firmer grounds\footnote{If we take the definition $e^{n\, h_{\mathrm{c}}(W)}$ of the complex typical volume for granted, this may imply that our arbitrary choice of the (multivalued) imaginary part of the logarithm in Eq.~\eqref{eq:def_complex_logarithm} is actually insignificant because it is the exponential of the complex Wigner entropy $e^{h_{\mathrm{c}}(W)}$ that is the physically relevant quantity.}.

On a more hypothetical note, the analytic continuation of the Shannon entropy functional in the complex plane suggests that one could go even one step further and consider the entropy of complex-valued functions. Here, the Wigner function is real-valued (albeit it can be both positive or negative) since the density operator is Hermitian, but we may possibly define the complex entropy of the complex-valued Weyl transform of non-Hermitian operators. All tools of complex analysis, such as the residue theorem, could possibly be used here. Yet, there remains to find a good interpretation of such a complex entropy.



\bigskip
\noindent \textbf{Note added:} The early steps of this work have been reported in  the Ph.D. thesis of one of us \cite{HertzPhD2018}, but the present paper presents a more complete analysis of complex entropies in phase space. After completion of this work, we have become aware of a few papers where the Shannon differential entropy associated with a Wigner function had been mentioned, although mostly via numerical investigations, for Wigner-positive \cite{Guevara2003-na} and Wigner-negative states \cite{Laguna2010-ir, Salazar2023-yf}.
Recently, building on the present work, the relative entropy version of the Wigner entropy (including its complex extension) has also been put forward as a means to measure the non-Gaussianity of a quantum state \cite{pizzimenti2023complexvalued}.

\bigskip
\noindent \textbf{Acknowledgments:} We would like to thank Michael G. Jabbour for useful discussions at an early stage of this work. 
N.J.C. is grateful to the James C. Wyant College of Optical Sciences for hospitality during his sabbatical leave in the autumn 2022, when this work was completed. Z.V.H. acknowledges support from the Belgian American Educational Foundation and from the Army Research Office (ARO) MURI Program Project on Quantum Network Science, "Theory and Engineering of Large-Scale Distributed Entanglement", awarded under grant number W911NF2110325. N.J.C. acknowledges support by the Fonds de la Recherche Scientifique – FNRS under Grant No T.0224.18 and by the European Union under project ShoQC within ERA-NET Cofund in Quantum Technologies (QuantERA) program.
A.H. acknowledges that the NRC headquarters is located on the traditional unceded territory of the Haudenosaunee and Mohawk people.

\bigskip
\noindent \textbf{Conflict of interest:}
The authors have no conflicts to disclose.

\bigskip
\noindent \textbf{Data availability:}
The data that support the findings of this study are available from the corresponding author upon reasonable request.

\bibliography{ComplexWigner}

\clearpage
\newpage
\appendix
\onecolumngrid
\section{Properties of real-valued extensions of the Wigner entropy}
\label{sec:real-valued}

\subsection{Numerical exploration of real-valued functionals}

In this Appendix, we numerically investigate the properties of the real-valued functionals defined in Sec. \ref{sec:alternative}.
For ease of reading, we recall here the expressions of the real entropy $h_{\mathrm{r}}$, the absolute entropy $h_{\mathrm{a}}$, the positive entropy $h_+$ and the negative entropy $h_-$:
\begin{align*}
    h_{\mathrm{r}}(W)
    =
    -\iint W(x,p)\ln\abs{W(x,p)}\dd x\dd p,
    &&
    h_{\mathrm{a}}(W)
    =
    -\iint\abs{W(x,p)}\ln\abs{W(x,p)}\dd x\dd p,
    \\
    h_+(W)
    =
    -\iint\limits_{W(x,p)\geq 0}W(x,p)\ln\abs{W(x,p)}\dd x\dd p,
    &&    
    h_-(W)
    =
    \iint\limits_{W(x,p)\leq 0}W(x,p)\ln\abs{W(x,p)}\dd x\dd p.
\end{align*}

In the following, we test whether $h_{\mathrm{r}}$, $h_{\mathrm{}}$ or $h_+$ are lower bounded by $\ln\pi+1$ (see property 5 in Sec. \ref{sec:alternative}) and if they satisfy the subadditivity property (\textit{i.e.}, $h_{r/a/+}(W)\leq h(\rho_x)+h(\rho_p)$, see property 4 in Sec. \ref{sec:alternative}).
We compare the values taken by the different functionals for a set of randomly generated pure states, namely, random superpositions of the first ten Fock states. 
We will only look at pure states here as we expect them to be the most likely to saturate uncertainty relations.
Figs. \ref{fig:real_entropy}, \ref{fig:positive_entropy}, \ref{fig:absolute_entropy} display the value taken by the above mentioned functionals (left side of each figure) and the value of the difference $h(\rho_x)+h(\rho_p)-h_{r/a/+}(W)$ (right side of each figure).
We plot each quantity as a function of the negative volume\footnote{We make this choice because the imaginary part of the complex-valued Wigner entropy is proportional to $\Vol_{-}(W)$, see Sec. \ref{sec:Complex-Wigner}.}, \textit{i.e.}, $\mathrm{Vol}_-(W)=-\iint_{W(x,p)<0}W(x,p)\dd x\dd p$.
The negative entropy $h_{-}$ is also plotted for illustration in Fig. \ref{fig:negative_entropy} although we do not expect it to be a possible extension of the Wigner entropy (nor to be lower bounded by $\ln\pi+1$). 
The functional $h_-$ will be a quantity of interest in the following of this Appendix (see Sec. \ref{sec:lower-bound}).

\subsubsection{Minimum value of the functionals}
Let us first look at the values taken by $h_{r/a/+}$ (left sides of Figs. \ref{fig:real_entropy}, \ref{fig:positive_entropy}, \ref{fig:absolute_entropy}) and see whether any of them is lower bounded by $\ln\pi+1$ (property 5). 
From the different graphs, it appears that $h_{\mathrm{r}}$ sometime takes values strictly lower than $\ln\pi+1$ (see Fig. \ref{fig:real_entropy}) whereas
both the functionals $h_{+}$ and $h_{\mathrm{}}$ seem to satisfy the expected lower bound $\ln\pi+1$ (see Figs. \ref{fig:positive_entropy} and  \ref{fig:absolute_entropy}).
However, since $h_{+}$ is always lower than or equal to $h_{\mathrm{}}$ [\textit{i.e.}, $h_{+}(W)\leq h_{\mathrm{}}(W)$, $\forall W\in\mathcal{W}$], the functional $h_{+}$ should give a tighter bound than $h_{a}$. Unfortunately, we have not been able to prove these bounds. 

It is interesting to note in Fig. \ref{fig:positive_entropy} the apparent almost-linear relation between the negative volume $\Vol_{-}(W)$ and the positive entropy $h_{+}(W)$ for pure states.
It is not very surprising to observe that, for pure states, $h_{+}(W)$ grows with $\mathrm{Vol}_-(W)$: indeed, the negative volume is known to be a witness of nonclassicality, while classical pure states (coherent states) minimize most uncertainty relations.
Nonclassical states are thus more likely to yield a higher uncertainty.

Moreover, we note in Figs. \ref{fig:positive_entropy}, \ref{fig:absolute_entropy}, and \ref{fig:negative_entropy} that Fock states generally seem to have the have the lowest possible value of $h_{+}$, $h_{a}$, and $h_{-}$ for a fixed value of the negative volume $\Vol_{-}(W)$. This is reminiscent to the known property that Fock states are minimum-uncertainty states for a given degree of non-Gaussianity \cite{PhysRevA.86.030102}. Indeed, by considering uncertainty relations  
under some Gaussianity constraint, it was shown in Ref. \cite{PhysRevA.86.030102} that saturation is achieved by all Fock states and some specific mixtures of subsequent Fock states. Here, the negative volume $\Vol_{-}(W)$ plays the role of such a measure of non-Gaussianity (which, for pure states, goes together with Wigner negativity). However, we have not been able to prove this property here.

\subsubsection{Comparison with the sum of the marginal entropies}
Let us now look at the values taken by $h(\rho_x)+h(\rho_p)-h_{r/+/a}$ (right sides of Figs. \ref{fig:real_entropy}, \ref{fig:positive_entropy}, \ref{fig:absolute_entropy}) to see if any of these functionals could be a lower bound on the sum of the marginal entropies (property 4).
In each case, it clearly appears that counterexamples can be found such that $h_{\mathrm{r}}$, $h_+$ or $h_{\mathrm{}}$ is not a lower bound on $h(\rho_x)+h(\rho_p)$.
The case of the real entropy $h_{\mathrm{r}}$ is the more subtle, as it is \textit{most of the time} a lower-bound on the sum of the marginal entropies, except for some rare counterexamples.
Also, the violation $h(\rho_x)+h(\rho_p)-h_{\mathrm{r}}(W)<0$ for these counterexamples is rather small.
This contrasts with the case of $h_+$ and $h_{\mathrm{}}$ which appear to be greater than $h(\rho_x)+h(\rho_p)$ in general.
Note however that by squeezing the Wigner function along an axis not aligned with $x$ or $p$, it is possible to increase arbitrarily the value of $h(\rho_x)+h(\rho_p)$ while keeping $h_+(W)$ and $h_{\mathrm{}}(W)$ fixed, so that neither $h_+$ or $h_{\mathrm{}}$ can be an upper bound to the sum of the marginal entropies.

\begin{figure}[H]
	\includegraphics[width=\textwidth]{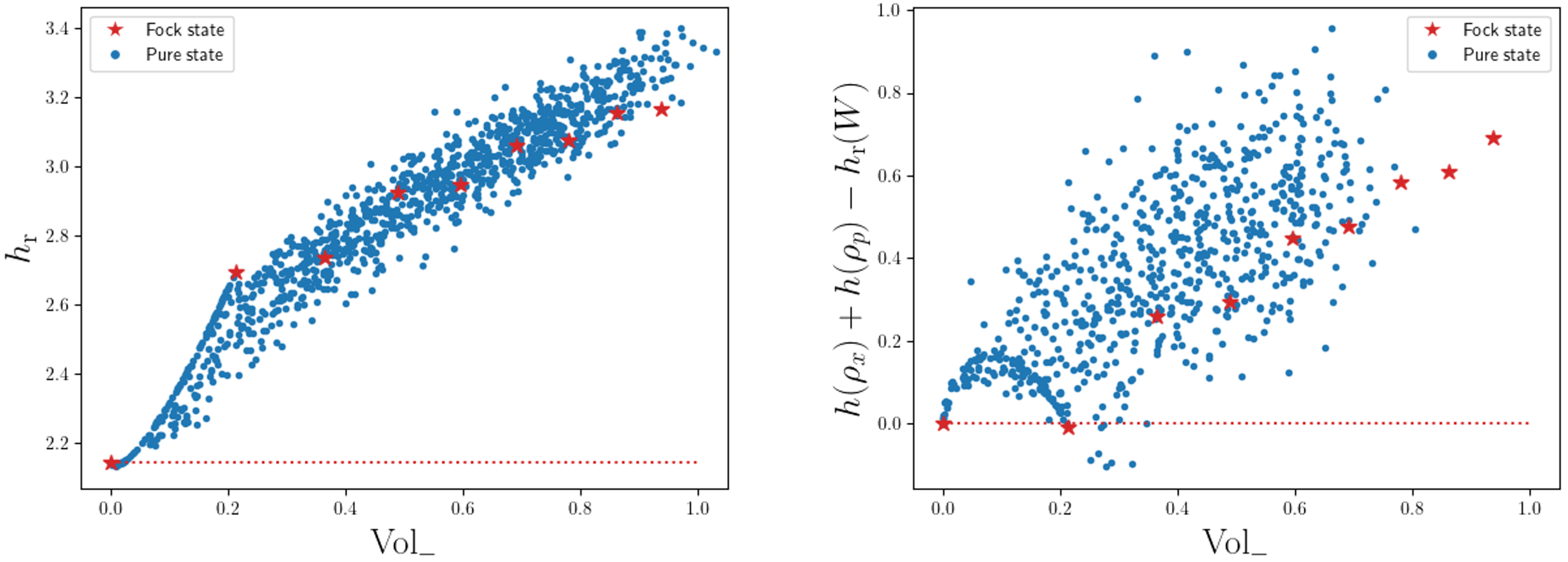}
	\caption{
	On the left, plot of the real Wigner entropy $h_{\mathrm{r}}(W)$ as a function of the negative volume $\Vol_{-}(W)$.
		Blue points are randomly generated pure states.
		Red stars are Fock states.
		As it appears, there exist some pure states which are such that $h_{r}(W)<\ln\pi+1$, contradicting property 5.
  On the right, plot of the difference $h(\rho_x)+h(\rho_p)-h_{\mathrm{r}}(W)$ for randomly generated pure states.
  Since the quantity is not always non-negative, we conclude that the real entropy $h_{\mathrm{r}}$ is not a lower bound on the sum of the marginal entropies, contradicting property 4.
	}
        \label{fig:real_entropy}
\end{figure}

\begin{figure}[H]
	\includegraphics[width=\textwidth]{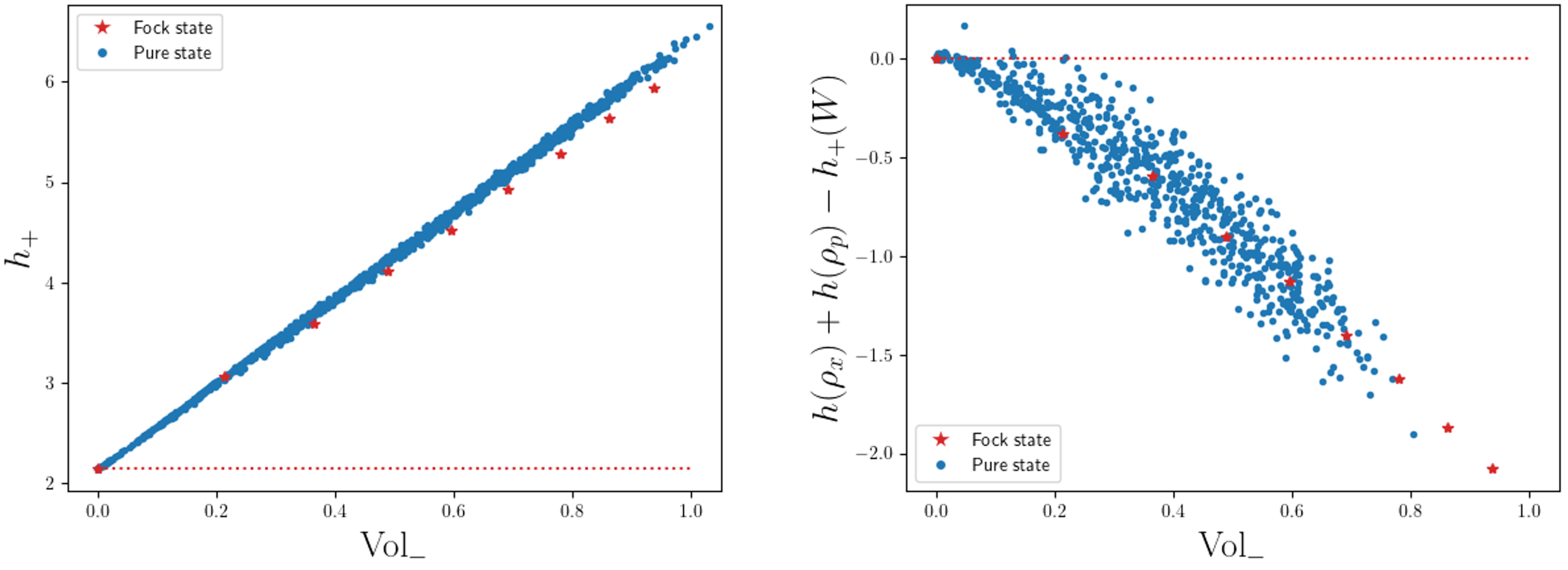}
	\caption{
		On the left, plot of the positive Wigner entropy $h_{+}(W)$ as a function of the negative volume $\Vol_{-}(W)$.  
		Blue points are randomly generated pure states.
		Red stars are Fock states.
		We observe empirically that pure states state obey a relation of the type $h_{+}(W)\approx \ln\pi+1+\kappa \Vol_{-}(W)$, with $\kappa\approx 4$.
  All of our randomly generated states are such that $h_+(W)\geq\ln\pi+1$, which suggests that the positive entropy satisfies property 5.
    On the right, plot of the difference $h(\rho_x)+h(\rho_p)-h_+(W)$ for randomly generated pure states.
  Since the quantity is sometime striclty negative, we conclude that the real entropy $h_{+}$ is not a lower bound on the sum of the marginal entropies, contradicting property 4.
	}
        \label{fig:positive_entropy}
\end{figure}

\begin{figure}[H]
\includegraphics[width=\textwidth]{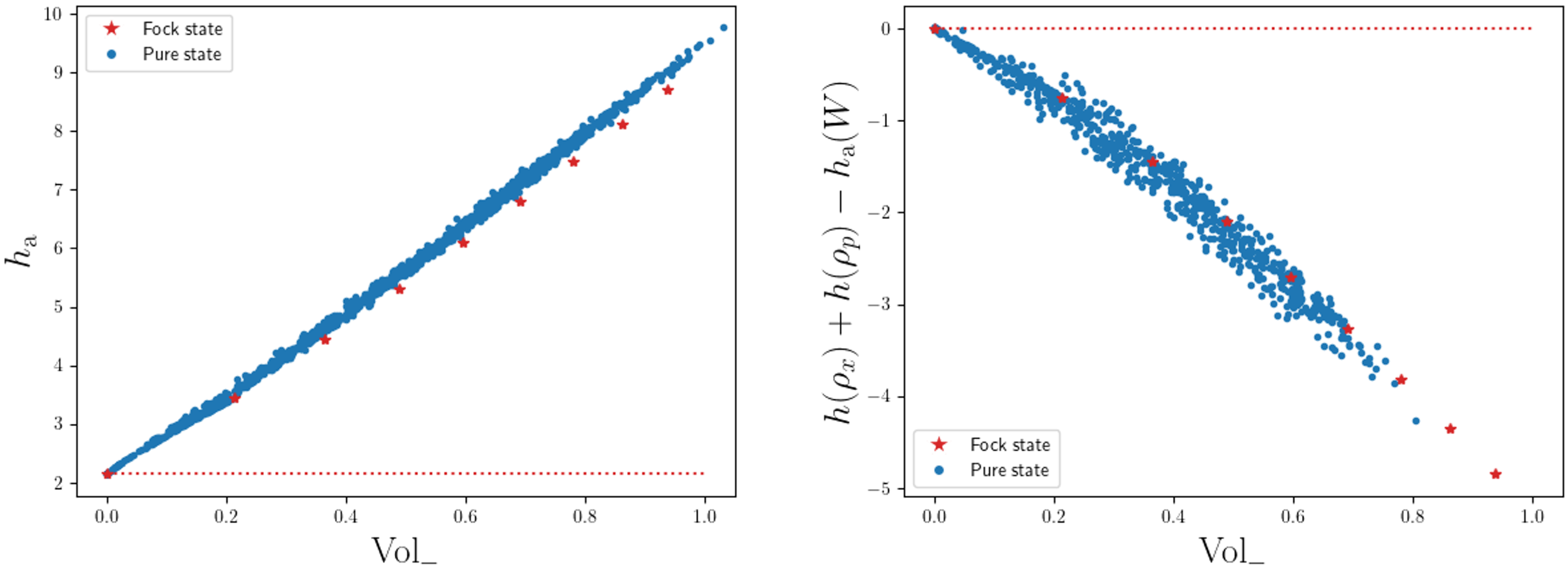}
\caption{
On the left, plot of the absolute Wigner entropy 
entropy $h_{a}(W)$ as a function of the negative volume  $\Vol_{-}(W)$.
Blue points are randomly generated pure states.
Red stars are Fock states.
All of our randomly generated states are such that $h_{\mathrm{}}(W)\geq\ln\pi+1$, which suggests that the absolute entropy satisfies property 5.
On the right, plot of the difference $h(\rho_x)+h(\rho_p)-h_{\mathrm{}}(W)$ for randomly generated pure states.
Since the quantity is not always non-negative, we conclude that the absolute entropy $h_+$ is not a lower bound on the sum of marginal entropies, contradicting property 4.
}
\label{fig:absolute_entropy}
\end{figure}

\begin{figure}[H]
\centering
\includegraphics[width=0.5\textwidth]{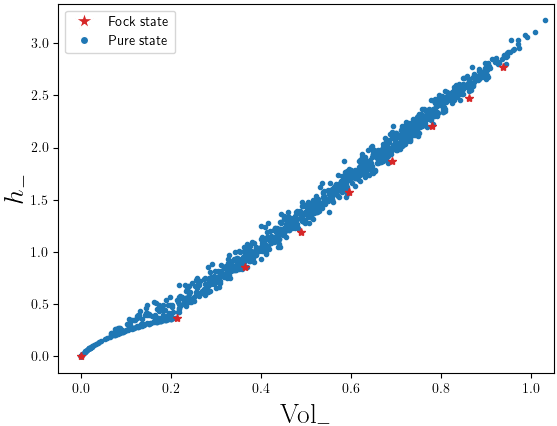}
\caption{
    Plot of the negative Wigner entropy $h_{-}(W)$ as a function of the negative volume $\Vol_{-}(W)$.
    Blue points are randomly generated pure states.
    Red stars are Fock states. 	
  }
    \label{fig:negative_entropy}
\end{figure}


\subsection{Relating the positive entropy to the sum of marginal entropies}
\label{sec:lower-bound}

Here, we make an attempt to show a relation between the positive entropy $h_{+}(W)$ and the sum of the marginal entropies $h({\rho_x})+h({\rho_p})$.
We saw in the previous section that the value $h_{+}(W)$ alone is not a lower bound on $h(\rho_x)+h(\rho_p)$, but it could be possible to find a correcting term $\Delta(W)$ such that $h(\rho_x)+h(\rho_p)\geq h_+(W)+\Delta(W)$.
We choose to focus on $h_+$ because its almost-linear relation with the negative volume suggests it is a quantity of interest.

\subsubsection{Splitting the positive and negative parts of the Wigner function}

Any Wigner function can be considered as the sum of its positive parts and negative parts.
With this in mind, we define $W^+$ and $W^-$ as follows:

\begin{equation*}
    W^+(x,p) = \max\lbrace W(x,p), 0\rbrace,
\qquad\qquad
    W^-(x,p) = -\min\lbrace W(x,p), 0\rbrace.
\end{equation*}
We can then write $W$ as $W(x,p)=W^+(x,p)-W^-(x,p)$.
The distribution $W^+(x,p)$ and $W^-(x,p)$ are non-negative and normalized to respectively $\Vol_+(W)$ and $\Vol_-(W)$.
Also, note that the positive entropy $h_{+}(W)$ simply corresponds to the entropy of the (non-normalized) distribution $W^+$, so that we have $h_{+}(W)=h(W^+)$.
Similarly, the negative entropy $h_{-}(W)$ is the entropy of the (non-normalized) distribution $W^{-}$ and we have $h_{-}(W)=h(W^{-})$.
We now define the marginal distributions of $W$, $W^+$ and $W^-$ along the $x$ and $p$ axis as follows:
\begin{align*}
	\rho_x(x)
	=
	\int W(x,p)\dd p,  &\qquad &\rho_p(p)
	=
	\int W(x,p)\dd x,
        \\
        \rho^+_x(x)
	=
	\int W^+(x,p)\dd p, &\qquad &  \rho^+_p(p)
	=
	\int W^+(x,p)\dd x,
        \\
        \rho^-_x(x)
	=
	\int W^-(x,p)\dd p, &\qquad &\rho^-_p(p)
	=
	\int W^-(x,p)\dd x.
\end{align*}
These different marginal distributions are related as $\rho_x(x) = \rho^+_x(x)-\rho^-_x(x)$ and $\rho_p(p) = \rho^+_p(p)-\rho^-_p(p)$.
Also, the marginal distributions are normalized according to the positive and negative volume: $\int \rho^+_x(x)\dd x=\int \rho^+_p(p)\dd p=\Vol_+(W)$ and $\int \rho^-_x(x)\dd x=\int \rho^-_p(p)\dd p=\Vol_-(W)$.

\subsubsection{Relating the marginal entropies}

Let us define the following marginal entropies:
\begin{align*}
	h(\rho_x)= -\int\rho_x(x)\ln\rho_x(x)\dd x, &\qquad & 	h(\rho_p)= -\int\rho_p(x)\ln\rho_p(p)\dd p,
        \\
        h(\rho^+_x)= -\int\rho^+_x(x)\ln\rho^+_x(x)\dd x,&\qquad &          h(\rho^+_p)= -\int\rho^+_p(p)\ln\rho^+_p(p)\dd p,
        \\
        h(\rho^-_x)= -\int\rho^-_x(x)\ln\rho^-_x(x)\dd x,&\qquad & h(\rho^-_p)= -\int\rho^-_p(p)\ln\rho^-_p(p)\dd p.
\end{align*}

We are going to derive a lower-bound on $h(\rho_x)$ as a function of $h(\rho_x^+)$ and $h(\rho_x^-)$.

\begin{align*}
	h(\rho_x)
	&=
	-\int
	\rho_x(x)
	\ln\rho_x(x)
	\dd x
	\\
	&=
	-\int
	\left(
	\rho^+_x(x)-\rho^-_x(x)
	\right)
	\ln
	\left(
	\rho^+_x(x)-\rho^-_x(x)
	\right)
	\dd x
	\\
	&\geq
	-\int
	\left(
	\rho^+_x(x)-\rho^-_x(x)
	\right)
	\ln
	\rho^+_x(x)
	\dd x
	\\
	&=
	-\int
	\rho^+_x(x)
	\ln
	\rho^+_x(x)
	\dd x
	+\int
	\rho^-_x(x)
	\ln
	\rho^+_x(x)
	\dd x
	\\
	&=
	h(\rho^+_x)
	+\int
	\rho^-_x(x)
	\ln
	\rho^+_x(x)
	\dd x
	\\
	&\geq
	h(\rho^+_x)
	+\int
	\rho^-_x(x)
	\ln
	\rho^-_x(x)
	\dd x
	\\
	&=
	h(\rho^+_x)-h(\rho^-_x)
\end{align*}

For both inequalities, we have used the monotonicity of logarithm and the fact that $\rho_x^+(x)\geq\rho_x^-(x)\geq 0$.
Note that the same development can be done for $h(\rho_p)$, so that we also have $h(\rho_p)\geq h(\rho_p^+)-h(\rho_p ^-)$.
In conclusion, it follows that:
\begin{align}
	h(\rho_x)+h(\rho_p)
	\geq
	h(\rho^+_x)+h(\rho^+_p)
	-h(\rho^-_x)-h(\rho^-_p).
 \label{eq:lower_bound_marg_entrop_with_pos_and_neg_marg}
\end{align}

\onecolumngrid
\subsubsection{Subadditivity of entropy for non-normalized distributions}

The subadditivity of entropy establishes an inequality between the entropy of a joint probability distribution and the entropy of its marginal distributions.
As we are going to show, it is possible to derive a similar relation for distributions which are non-negative but not normalized to $1$. 
Let us consider the function $W^+(x,p)$ which is non-negative and normalized to $\Vol_+(W)$.
It is associated with the marginal distributions $\rho_x^+(x)$ and $\rho_p^+(p)$.
Of course, keep in mind that the entropy of $W^{+}$ is the positive entropy of $W$, \textit{i.e.}, $h(W^+)=h_+(W)$.
In the following development, we will refer to it as $h(W^+)$ because we focus on the distribution $W_+$.

Ref. \cite{CoverThomas} provides a simple proof based on Jensen's inequality for the subbadditivity of entropy.
We follow the same lines, and generalize hereafter the proof to non-normalized distributions.
The entropy of $W^+$ is written as:
\begin{align}
	h(W^+)
	=
	-\iint W^+(x,p)
	\ln W^+(x,p)
	\dd x\dd p,
\end{align}
and the sum of the entropies of $\rho_x^+$ and $\rho_p^+$ is:
\begin{align*}
	h(\rho^+_x)+h(\rho^+_p)
	&=
	-\int
	\rho^+_x(x)\ln\rho^+_x(x)\dd x
	-\int\rho^+_p(p)\ln\rho^+_p(p)\dd p
	\\
	&=
	-\iint W^+(x,p)\dd p
	\ln\rho^+_x(x)\dd x
	-\iint W^+(x,p)\dd x
	\ln\rho^+_p(p)\dd p
	\\
	&=
	-\iint W^+(x,p)
	\left(
	\ln(\rho^+_x(x))
	+
	\ln(\rho^+_p(p))
	\right)
	\dd x\dd p
	\\
	&=
	-\iint W^+(x,p)
	\ln
	\left(
	\rho^+_x(x)
	\rho^+_p(p)
	\right)
	\dd x\dd p.
\end{align*}
We can then compute the difference of these two quantities as follows:
\begin{align*}
	h(\rho^+_x)+h(\rho^+_p)-h(W^+)
	&=
	-\iint W^+(x,p)
	\ln
	\left(
	\rho^+_x(x)
	\rho^+_p(p)
	\right)
	\dd x\dd p
	+
	\iint
	W^+(x,p)
	\ln
	W^+(x,p)
	\dd x\dd p
	\\
	&= 
	-\iint W^+(x,p)
	\left(
	\ln
	\left(
	\rho^+_x(x)
	\rho^+_p(p)
	\right)
	-\ln W^+(x,p)
	\right)
	\dd x\dd p
	\\
	&=
	-\iint W^+(x,p)
	\ln
	\left(
	\frac{\rho^+_x(x)\rho^+_p(p)}{W^+(x,p)}
	\right)
	\dd x\dd p
	\\
	&=
	-\Vol_+(W)
	\iint 
	\frac{W^+(x,p)}{\Vol_+(W)}
	\ln
	\left(
	\frac{\rho^+_x(x)\rho^+_p(p)}{W^+(x,p)}
	\right)
	\dd x\dd p
	\\
	&\geq
	-\Vol_+(W)
	\ln\left(
	\iint 
	\frac{W^+(x,p)}{\Vol_+(W)}
	\frac{\rho^+_x(x)\rho^+_p(p)}{W^+(x,p)}
	\dd x\dd p
	\right)
	\\
	&=
	-\Vol_+(W)
	\ln\left(
	\iint 
	\frac{\rho^+_x(x)\rho^+_p(p)}{\Vol_+(W)}
	\dd x\dd p
	\right)
	\\
	&=
	-\Vol_+(W)
	\ln\left(
	\frac{1}{\Vol_+(W)}
	\int 
	\rho^+_x(x)
	\dd x
	\int 
	\rho^+_p(p)
	\dd p
	\right)
	\\
	&=
	-\Vol_+(W)
	\ln\left(
	\frac{\Vol_+(W)^2}{\Vol_+(W)}
	\right)
	\\
	&=
	-\Vol_+(W)\ln\Vol_+(W)
\end{align*}
where the inequality follows from Jensen's inequality.
A similar development can be done for $h(W^-)$, which proves the two following relations:
\begin{align}
	h(\rho^+_x)+h(\rho^+_p)
	\geq
	h(W^+)-\Vol_+(W)\ln\Vol_+(W),
        \label{eq:lower_bound_hrho+}
	\\ \nonumber
	\\
	h(\rho^-_x)+h(\rho^-_p)
	\geq
	h(W^-)-\Vol_-(W)\ln\Vol_-(W).
         \label{eq:lower_bound_hrho-}
\end{align}

\subsubsection{Expression of a lower-bound on the sum of the marginal entropies}

Remember that, from Eq. \eqref{eq:lower_bound_marg_entrop_with_pos_and_neg_marg}, we would like to find a lower bound to the quantity $h(\rho_x^+)+h(\rho_p^+)-h(\rho_x^-)-h(\rho_p^-)$.
To do so, it would be sufficient to lower bound $h(\rho_x^+)+h(\rho_p^+)$ and upper bound $h(\rho_x^-)+h(\rho_p^-)$.
We can use Eq. \eqref{eq:lower_bound_hrho+} to lower bound $h(\rho_x^+)+h(\rho_p^+)$.
However, Eq. \eqref{eq:lower_bound_hrho-} provides us with a lower bound on $h(\rho_x^-)+h(\rho_p^-)$, which is not what we are looking for.

If we combine our precedent observations, we can write the following:
\begin{align}
	h(\rho_x)+h(\rho_p)
	&\geq
	h(\rho^+_x)+h(\rho^+_p)
	-h(\rho^-_x)-h(\rho^-_p)
	\\
	&\geq
	h(W^+)-\Vol_+(W)\ln\Vol_+(W)
	-h(\rho^-_x)-h(\rho^-_p)
	\\
	&=
	\underbrace{h_{+}(W)-\Vol_+(W)\ln\Vol_+(W)
	}_{\text{symmetric}}
	\underbrace{
	-h(\rho_x^-)-h(\rho_p^-)
	}_{\text{not symmetric}}
 \label{eq:analytical_lower_bound}
\end{align}

The lower bound that we have derived on the sum of the marginal entropies is constituted of two main terms.
The first term $h_{+}(W)-\Vol_+(W)\ln\Vol_+(W)$ is the sum of the positive entropy and a function of the positive volume.
That first term is thus symmetric, \textit{i.e.}, invariant under area-preserving transformation.
However, the second term $-h(\rho_x^-)-h(\rho_p^-)$ is not symmetric.
Indeed, the sum of the marginal entropies of a distribution is not invariant under area-preserving transformations.
For example, by squeezing the distribution $W^-$ along an axis that is not aligned with either the $x$ or $p$ axis, the value of $h(\rho_x^-)+h(\rho_p^-)$ can be set to an arbitrarily high value.

In conclusion, even if Eq. \eqref{eq:analytical_lower_bound} is an encouraging result as it provides an analytical lower bound on the sum of the marginal entropies for \textit{any} Wigner function (possibly taking negative values), work still remains to be done in order to make that lower bound completely symmetric.

\section{Complex Wigner entropy of Schrödinger cat states at the large-amplitude limit}
\label{sec:cat-states-high-alpha}

Any binary superposition of coherent states can be converted to the following canonical expression through some Gaussian unitary (displacement and rotation):
\begin{equation}
	\ket{\psi}
	=
	\frac{1}{\sqrt{C}}
	\left(
	\sqrt{m}
	\ket{\alpha}+
	e^{i\varphi}
	\sqrt{1-m}
	\ket{-\alpha}	
	\right),
        \label{eq:ket_cat_state_canonical}
\end{equation}
where $\alpha\in\mathbb{R}_+$, $\varphi\in[0, 2\pi)$ and $m\in [0,1]$.
$C$  is computed to $C=1+2\sqrt{m(1-m)}\exp(-2\alpha^2)\cos(\varphi)$.
Let us insist on the fact that we consider the parameter $\alpha$ to be real-valued and non-negative.
The projector of the state $\ket{\psi}$ is then:
\begin{align}
	\ket{\psi}\bra{\psi}
	=
	\frac{1}{C}\left(
	m\ket{\alpha}\bra{\alpha}
	+
	(1-m)\ket{-\alpha}\bra{-\alpha}
	+
	\sqrt{m(1-m)}e^{i\varphi}\ket{-\alpha}\bra{\alpha}
	+
	\sqrt{m(1-m)}e^{-i\varphi}\ket{\alpha}\bra{-\alpha}
	\right)
 \label{eq:projector_canonical_cat_state}
\end{align}
In this Appendix, we are going to compute the Wigner function of $\ket{\psi}\bra{\psi}$.
Then, we will compute the real and imaginary parts of the complex Wigner entropy in the limit of large $\alpha$.

\subsection{Wigner function}

Let us define $W_{\alpha\beta}(x,p)$ has the Wigner function of the operator $\ket{\alpha}\bra{\beta}$ (in the coherent basis).
The expression of $W_{\alpha\beta}$ can be found in \cite{ZVHPhD2018} and reads as follows:
\begin{align}
W_{\alpha\beta}(x, p)=\frac{1}{\pi} \exp \left(-\frac{1}{2}\left(|\alpha|^2+|\beta|^2-2 \alpha \beta^*\right)-\left(x-\frac{\alpha+\beta^*}{\sqrt{2}}\right)^2-\left(p-\frac{\alpha-\beta^*}{\sqrt{2} i}\right)^2\right)
\end{align}

We can then use the above expression to compute the Wigner function of the operators $\ket{\alpha}\bra{\alpha}$, $\ket{-\alpha}\bra{-\alpha}$, $\ket{\alpha}\bra{-\alpha}$ and $\ket{-\alpha}\bra{\alpha}$.
The Wigner function of $\ket{\psi}\bra{\psi}$ will then be the corresponding mixture (with complex coefficients) as defined from Eq. \eqref{eq:projector_canonical_cat_state}.
Note that since $\ket{\psi}\bra{\psi}$ is Hermitian, its Wigner function is real-valued.
This yields the following expression:
\begin{align}
    W(x,p)
    =&
    \frac{1}{\pi C}\Bigg(
    m\exp(-(x-\sqrt{2}\alpha)^2-p^2)
    +
    (1-m)\exp(-(x+\sqrt{2}\alpha)^2-p^2)
    \nonumber
    \\&+
    2\sqrt{m(1-m)}
    \exp(-x^2-p^2)
    \cos\left(
    2\sqrt{2}\alpha p
    +\varphi
    \right)
    \Bigg)
    \\
    =&
    \frac{m}{C}
    W_0\left(x-\sqrt{2}\alpha,p\right)
    +
    \frac{1-m}{C}
    W_0\left(x+\sqrt{2}\alpha,p\right)
    +\frac{2\sqrt{m(1-m)}}{C}
    W_0(x,p)
    \cos(2\sqrt{2}\alpha p+\varphi)
    \label{eq:wig_cat_as_function_of_w0}
\end{align}
where in the second expression we have used the Wigner function of vacuum $W_0(x,p)=\exp(-x^2-p^2)/\pi$.

\subsection{Complex Wigner entropy in the large $\alpha$ approximation}
\label{subsec:compl_wigent_in_large_alpha_approx}

Eq. \eqref{eq:wig_cat_as_function_of_w0} is the exact Wigner function of a state $\ket{\psi}$ in the form of Eq. \eqref{eq:ket_cat_state_canonical}.
At this point, we will now make the assumption that $\alpha\gg 1$.
This implies then that $C\approx 1$.
We can then approximate Eq. \eqref{eq:wig_cat_as_function_of_w0} as follows:
\begin{equation}
	W(x,p)
	\approx
	\underbrace{
	m
	W_0\left(x-\sqrt{2}\alpha,p\right)
	}_{W_A}
	+
	\underbrace{
	(1-m)
	W_0\left(x+\sqrt{2}\alpha,p\right)
	}_{W_B}
	+
	\underbrace{
	2\sqrt{m(1-m)}
	W_0(x,p)
	\cos(2\sqrt{2}\alpha p+\varphi)
	}_{W_C}
\end{equation}

Observe that the above expression contains three distinct terms.
Notice that the sum of the two first term ($W_A+W_B$) is the Wigner function of the mixture of coherent states $\hat{\rho}= m\ket{\alpha}\bra{\alpha}+(1-m)\ket{-\alpha}\bra{-\alpha}$.
We can then interpret the term $W_C$ as the interference term of the cat state.

The term $W_A$ is located around the point $(x=\sqrt{2}\alpha, p=0)$, the term $W_B$ is located around the the point $(x=-\sqrt{2}\alpha, p=0)$ and the term $W_C$ around the point $(x=0,p=0)$.
In the large $\alpha$ approximation, the distributions $W_A(x,p)$, $W_B(x,p)$ and $W_C(x,p)$ are thus well separated over phase space, so that we may consider that they do not overlap.
In other words, we can make the approximation that whenever one takes non-zero values, then the two others are assumed to be zero.
We will use that approximation to compute the complex Wigner entropy of the cat state.

\begin{align}
	h(W)
        &=
        h(W_A+W_B+W_C)
	\\&=
	-\iint W\ln W\dd x\dd p
	\\&=
	-\iint\left(W_A+W_B+W_C\right)\ln\left(W_A+W_B+W_C\right)\dd x\dd p
	\\&\approx
	-\iint\limits_{W_A\neq 0} W_A\ln W_A\dd x\dd p
	-\iint\limits_{W_B\neq 0} W_B\ln W_B\dd x\dd p
	-\iint\limits_{W_C\neq 0} W_C\ln W_C\dd x\dd p
	\\&=
	h(W_A)+h(W_B)+h(W_C)
\end{align}

For distributions that do not overlap, the entropy of the sum of the distributions is the sum of the entropy of each distribution.
Using the fact that $h(W_0)=\ln\pi+1$ and that $h(aW)=ah(W)-a\ln a$ (for $a>0$), we find that $h(W_A)=m(\ln\pi+1)-m\ln m$ and $h(W_B)=(1-m)(\ln\pi+1)-(1-m)\ln(1-m)$.
We have then:
\begin{align}
    h(W_A+W_B)
    &\approx h(W_A)+h(W_B)+h(W_C)
    \\&=
    \ln\pi+1-m\ln m-(1-m)\ln(1-m)
    \\&=
    h(W_0)+h_2(m)
\end{align}
where $h_2(m)\vcentcolon=-m\ln m-(1-m)\ln(1-m)$ is the binary entropy of $m\in[0,1]$.
As we noticed previously, the entropy of $h(W_A+W_B)$ is simply the Wigner entropy of a mixture of coherent states. 
Let us now focus on the entropy of the interference term, \textit{i.e.} $h(W_C)$.

Starting from the expression of $W_C$:
\begin{align}
	W_C(x,p)
	=
	\frac{2}{\pi}\sqrt{m(1-m)}
	\exp(-x^2)\exp(-p^2)\cos(2\sqrt{2}\alpha p+\varphi),
\end{align}
we compute the entropy of $W_C$ as follows:
\begin{align}
    h(W_C)
    =&
    -\iint
    \frac{2}{\pi}\sqrt{m(1-m)}
    \exp(-x^2)\exp(-p^2)\cos(2\sqrt{2}\alpha p+\varphi)
    \nonumber
    \\
    &\times
    \ln\left(
    \frac{2}{\pi}\sqrt{m(1-m)}
    \exp(-x^2)\exp(-p^2)\cos(2\sqrt{2}\alpha p+\varphi)
    \right)\dd x\dd p
    \\
    =&
    -\iint
    \frac{2}{\pi}\sqrt{m(1-m)}
    \exp(-x^2)\exp(-p^2)\cos(2\sqrt{2}\alpha p+\varphi)
    \nonumber
    \\
    &\times
    \left[
    \ln\left(\frac{2}{\pi}\sqrt{m(1-m)}\right)-x^2-p^2+\ln\cos(2\sqrt{2}\alpha p+\varphi)
    \right]
    \dd x\dd p
    \\
    =&
    -\frac{2}{\pi}\sqrt{m(1-m)}\ln\left(\frac{2}{\pi}\sqrt{m(1-m)}\right)
    \underbrace{\int\exp(-x^2)\dd x}_{\sqrt{\pi}}
    \underbrace{\int\exp(-p^2)\cos(2\sqrt{2}\alpha p+\varphi)\dd p}_{\exp(-2\alpha^2)\sqrt{\pi}\cos\varphi}
    \nonumber
    \\&
    +\frac{2}{\pi}\sqrt{m(1-m)}
    \underbrace{\int x^2 \exp(-x^2)\dd x}_{\frac{1}{2}\sqrt{\pi}}
    \underbrace{\int\exp(-p^2)\cos(2\sqrt{2}\alpha p+\varphi)\dd p}_{\exp(-2\alpha^2)\sqrt{\pi}\cos\varphi}
    \nonumber
    \\&
    +\frac{2}{\pi}\sqrt{m(1-m)}\int 
    \underbrace{\exp(-x^2)\dd x}_{\sqrt{\pi}}
    \underbrace{\int p^2\exp(-p^2)\cos(2\sqrt{2}\alpha p+\varphi)\dd p}_{\frac{1}{2}(1-4\alpha^2)\exp(-2\alpha^2)\sqrt{\pi}\cos\varphi}
    \nonumber
    \\&
    -\frac{2}{\pi}\sqrt{m(1-m)}
    \underbrace{\int\exp(-x^2)\dd x}_{\sqrt{\pi}}
    \;\exp(-p^2)\ln\cos(2\sqrt{2}\alpha p+\varphi)
    \cos(2\sqrt{2}\alpha p+\varphi)\dd p
    \\=&
    \sqrt{m(1-m)}\exp(-2\alpha^2)\cos\varphi
    \left(
    -2\ln\left(\frac{2}{\pi}\sqrt{m(1-m)}\right)+2-4\alpha^2	
    \right)
    \nonumber
    \\
    &
    -\frac{2}{\sqrt{\pi}}\sqrt{m(1-m)}
    \int\exp(-p^2)\ln\cos(2\sqrt{2}\alpha p+\varphi)
    \cos(2\sqrt{2}\alpha p+\varphi)\dd p
    \label{eq:wigner_entropy_interference_term}
\end{align}

Now, in the limit $\alpha\rightarrow\infty$, we have $\exp(-2\alpha^2)\rightarrow 0$.
In the large $\alpha$ approximation, only the second term of Eq. \eqref{eq:wigner_entropy_interference_term} survives and we have:
\begin{align}
    h(W_C)\approx
    -\frac{2}{\sqrt{\pi}}\sqrt{m(1-m)}
    \int
    \exp(-p^2)
    \ln\cos(2\sqrt{2}\alpha p+\varphi)
    \cos(2\sqrt{2}\alpha p+\varphi)\dd p.
    \label{eq:h_wc_approx_second_term}
\end{align}

In the following, we are going to split the real and imaginary parts of $h(W_C)$.
We will use the relation $\ln z=\ln\abs{z}+i\arg(z)$.
Also, we set $q=2\sqrt{2}\alpha p+\varphi$, so that $p=(q-\varphi)/(2\sqrt{2}\alpha)$.
We have then $\dd p=\dd q/(2\sqrt{2}\alpha)$.
Using this and starting from Eq. \eqref{eq:h_wc_approx_second_term}, we find:
\begin{align}
    h(W_C)
    \approx&
    -\frac{2}{\sqrt{\pi}}\sqrt{m(1-m)}
    \int\ln\cos(q)
    \exp(-\left(\frac{q-\varphi}{2\sqrt{2}\alpha}\right)^2)\cos(q)\frac{\dd q}{2\sqrt{2}\alpha}
    \\=&
    -\frac{1}{\alpha}\sqrt{\frac{m(1-m)}{2\pi}}
    \int
    \exp(-\left(\frac{q-\varphi}{2\sqrt{2}\alpha}\right)^2)
    \cos(q)\ln\cos(q)\dd q
    \\=&
    -\frac{1}{\alpha}\sqrt{\frac{m(1-m)}{2\pi}}
    \int
    \exp(-\left(\frac{q-\varphi}{2\sqrt{2}\alpha}\right)^2)
    \cos(q)\ln\abs{\cos(q)}\dd q
    \nonumber
    \\&
    -\frac{i}{\alpha}\sqrt{\frac{m(1-m)}{2\pi}}
    \int
    \exp(-\left(\frac{q-\varphi}{2\sqrt{2}\alpha}\right)^2)
    \cos(q)\arg(\cos(q))\dd q
    \label{eq:h_wc_real_and_imag}
\end{align}

We are going to evaluate the real part and the imaginary part of this integral separately.
To do so, we are going to split the domain of integration of the variable $p$ (which is $\mathbb{R}$) as an infinite countable union of subsets of length $2\pi$.
We divide the real line as follows: $\mathbb{A}_k=\big[2k\pi, 2(k+1)\pi\big)$, which is such that $\bigcup_{k\in\mathbb{Z}}\mathbb{A}_k=\mathbb{R}$.

\subsubsection{Real part}
Starting from Eq. \eqref{eq:h_wc_real_and_imag}, we can express the real part of $h(W_C)$ (in the large $\alpha$ approximation) as follows:
\begin{align}
    \mathrm{Re}\left[h(W_C)\right]
    \approx&
    -\frac{1}{\alpha}\sqrt{\frac{m(1-m)}{2\pi}}
    \int
    \exp(-\left(\frac{q-\varphi}{2\sqrt{2}\alpha}\right)^2)
    \cos(q)\ln\abs{\cos(q)}\dd q
    \\=&
    -\frac{1}{\alpha}\sqrt{\frac{m(1-m)}{2\pi}}
    \sum_{k\in\mathbb{Z}}\;
    \int_{\mathbb{A}_k}
    \exp(-\left(\frac{q-\varphi}{2\sqrt{2}\alpha}\right)^2)
    \cos(q)\ln\abs{\cos(q)}\dd q
\end{align}

When $\alpha$ is very large, the Gaussian $\exp(-((q-\varphi)/(2\sqrt{2}\alpha))^2)$ varies very slowly, so that we can assume that on a single interval $\mathbb{A}_k$ it keeps the same value.
The middle value of the interval $\mathbb{A}_k$ is $(2k+1)\pi$.
We can then write:
\begin{align}
    \mathrm{Re}\left[h(W_C)\right]
    \approx&
    -\frac{1}{\alpha}\sqrt{\frac{m(1-m)}{2\pi}}
    \sum\limits_{k\in\mathbb{Z}}
    \exp(-\left(\frac{(2k+1)\pi-\varphi}{2\sqrt{2}\alpha}\right)^2)
    \underbrace{\int\limits_{\mathbb{A}_k}
    \cos(q)\ln\abs{\cos(q)}\dd q}_{=0}
    \\
    &=
    0
\end{align}
where we have use the fact that the function $\phi(x)=x\ln(\abs{x})$ is odd, so that $\phi(\cos(x))$ is on average equal to zero over a period of cosine.

\subsubsection{Imaginary part}
We proceed similarly for the imaginary part.

\begin{align}
    \mathrm{Im}\left[h(W_C)\right]
    \approx&
    -\frac{1}{\alpha}\sqrt{\frac{m(1-m)}{2\pi}}
    \int
    \exp(-\left(\frac{q-\varphi}{2\sqrt{2}\alpha}\right)^2)
    \cos(q)\arg(\cos(q))\dd q
    \\=&
    -\frac{1}{\alpha}\sqrt{\frac{m(1-m)}{2\pi}}
	\sum\limits_{k\in\mathbb{Z}}\;
	\int\limits_{\mathbb{A}_k}
	\exp(-\left(\frac{q-\varphi}{2\sqrt{2}\alpha}\right)^2)
	\cos(q)\arg(\cos(q))\dd q
    \\\approx&
    -\frac{1}{\alpha}\sqrt{\frac{m(1-m)}{2\pi}}
	\sum\limits_{k\in\mathbb{Z}}
	\exp(-\left(\frac{(2k+1)\pi-\varphi}{2\sqrt{2}\alpha}\right)^2)
	\int\limits_{\mathbb{A}_k}
	\cos(q)\arg(\cos(q))\dd q
	\\
	=&
    -\frac{1}{\alpha}\sqrt{\frac{m(1-m)}{2\pi}}
    \underbrace{\sum\limits_{k\in\mathbb{Z}}
	\exp(-\left(\frac{(2k+1)\pi-\varphi}{2\sqrt{2}\alpha}\right)^2)}_{\sqrt{\frac{2}{\pi}}\alpha\vartheta_3\left(\frac{\pi-\varphi}{2},\exp(-2\alpha^2)\right)}
	\underbrace{\int\limits_{\frac{\pi}{2}}^{\frac{3\pi}{2}}
	\cos(q)\pi\dd q}_{-2\pi}
    \\
    =&
    2\sqrt{m(1-m)}\;\vartheta_3\left(\frac{\pi-\varphi}{2},\exp(-2\alpha^2)\right),
\end{align}
where $\vartheta_3(u,q)=1+\sum_{n=1}^{\infty}q^{n^2}\cos(2nu)$ is the third Jacobi elliptic theta function.
Observe that $\vartheta_3(u, 0)=1$, so that in the limit $\alpha\rightarrow\infty$ we have $\mathrm{Im}\left[h(W_C)\right]\rightarrow 2\sqrt{m(1-m)}$.

Let us now conclude.
In the limit of $\alpha\gg 1$, we have found that $\mathrm{Re}\left[h(W_C)\right]\rightarrow 0$ and $\mathrm{Im}\left[h(W_C)\right]\rightarrow 2\sqrt{m(1-m)}$.
Since we also know that in the limit $\alpha\gg 1$, we can write $h_{\mathrm{c}}(W)\approx h_{\mathrm{c}}(W_A)+h_{\mathrm{c}}(W_B)+h_{\mathrm{c}}(W_C)$, we have the following:
\begin{align}
	h_{\mathrm{c}}(W)
	&\approx
	\ln\pi+1-m\ln m-(1-m)\ln(1-m)
	+2i\sqrt{m(1-m)}
        \\
        &=
        h_{\mathrm{c}}(W_0)+h_2(m)+2i\sqrt{m(1-m)}
        \label{eq:compl_wig_entr_cat_state}
\end{align}
where $W_0$ is the Wigner function of vacuum and $h_2$ is the binary entropy.
Observe that Eq. \eqref{eq:compl_wig_entr_cat_state} does not depend on $\varphi$.
Also, the imaginary part is maximized for $m=1/2$, so that for a balanced superposition of coherent states (cat state), we find the following expression:
\begin{align}
	h(W_\mathrm{cat})
	\approx
	\ln\pi+1+\ln 2
	+i,
        \label{eq:wigent_catstate}
\end{align}
where $W_{\mathrm{cat}}$ is the Wigner function of a balanced cat state $\ket{\psi_\mathrm{cat}}\propto\ket{\alpha}+e^{i\varphi}\ket{-\alpha}$ (for any phase $\varphi\in[0, 2\pi)$).

\subsection{Generalization to a balanced superposition of many coherent states}

Interestingly, the expression we obtained for the complex Wigner entropy of a cat state (balanced superposition of two coherent states), see Eq. \eqref{eq:wigent_catstate}, can be generalized to a balanced superposition of $N$ coherent states.
Let us consider the following pure state:
\begin{align}
    \ket{\psi}
    =
    \frac{1}{\sqrt{C_{\bm{\alpha},\bm{\varphi}}}}
    \sum\limits_{n=1}^{N}
    e^{i\varphi_n}
    \ket{\alpha_i}
\end{align}
which is defines from a coherent amplitudes vector $\bm{\alpha}\in\mathbb{C}^{N}$ and a phase vector $\bm{\varphi}\in[0,2\pi)^{N}$.
The normalization constant $C_{\bm{\alpha},\bm{\varphi}}\in\mathbb{R}_+$ depends on both $\bm{\alpha}$ and $\bm{\varphi}$.
The density operator associated to that state is:
\begin{align}
    \ket{\psi}\bra{\psi}
    &=
    \frac{1}{C_{\bm{\alpha},\bm{\varphi}}}
    \sum\limits_{n=1}^{N}\sum\limits_{m=1}^{N}
    e^{i(\varphi_n-\varphi_m)}
    \ket{\alpha_n}\bra{\alpha_m},
    \\
    &=
    \frac{1}{C_{\bm{\alpha},\bm{\varphi}}}
    \sum\limits_{n}
    \ket{\alpha_n}\bra{\alpha_n}
    +
    \frac{1}{C_{\bm{\alpha},\bm{\varphi}}}
    \sum\limits_{n\neq m}
    e^{i(\varphi_n-\varphi_m)}
    \ket{\alpha_n}\bra{\alpha_m},
    \\
    &=
    \underbrace{
    \frac{1}{C_{\bm{\alpha},\bm{\varphi}}}
    \sum\limits_{n}
    \ket{\alpha_n}\bra{\alpha_n}
    }_{\textrm{mixture}}
    +
    \underbrace{
    \frac{1}{C_{\bm{\alpha},\bm{\varphi}}}
    \sum\limits_{m<n}
    \left[
    e^{i(\varphi_n-\varphi_m)}
    \ket{\alpha_n}\bra{\alpha_m}
    +h.c.
    \right]
    }_{\textrm{interference}}.
\end{align}

Observe that there are $N$ mixture terms and $N(N-1)/2$ interference terms.
In the approximation where none of the mixture terms and interference terms overlap, we will be able to use a similar reasoning as we did in \ref{subsec:compl_wigent_in_large_alpha_approx}.
There are in total $N(N+1)/2$ terms (mixture and interference) that shouldn't overlap for the approximation to hold.
This amount to checking $(N-1)N(N+1)(N+2)/8$ conditions (distances between each terms).

Notice now that every mixture term is located around the phase-space location $\alpha_n$ (for $n=1,..,N$), and every interference term is located at the location $(\alpha_n+\alpha_m)/2$ (for $n=1,...,N$ and $m=1,..,n-1$).
We define the matrix $\mathbf{A}\in\mathbb{C}^{N\times N}$ with components $A_{mn}=(\alpha_m+\alpha_n)/2$.
The conditions that have to be satisfied are the following,
\begin{align}
    \abs{
    A_{mn}-A_{pq}
    }\gg 1
    \qquad\forall 
    m\geq n> p\geq q,
\end{align}
or equivalently,
\begin{align}
    \abs{
    \alpha_m+\alpha_n-\alpha_p-\alpha_q
    }\gg 1
    \qquad\forall 
    m\geq n> p\geq q.
\end{align}

When the above approximation holds, we can compute the total complex Wigner entropy of $\ket{\psi}$ as the sum of the complex entropy of each terms.
Following the same lines as in \ref{subsec:compl_wigent_in_large_alpha_approx}, we find that the normalization constant $C_{\bm{\alpha},\bm{\varphi}}$ then tends to $N$.
The complex Wigner entropy of the sum of the mixture terms becomes $h(W_0)+\ln N$ (no imaginary part).
The complex Wigner entropy of the operator $\ket{\alpha_m}\bra{\alpha_n}$ (where $\abs{\alpha_m-\alpha_n}\gg 1$) is $2i$ (no real part), so that the complex Wigner entropy of the sum of the $N(N-1)/2$ interference terms is $2i\times N(N-1)/(2N)=i(N-1)$.
This yields the following approximation for the complex Wigner entropy of $\ket{\psi}$:
\begin{align}
h_{\mathrm{c}}(W_{\psi})
\approx
\underbrace{
\ln\pi+1+\ln N}_{\textrm{mixture}}
+
\underbrace{i\, (N-1)}_{\textrm{interference}}.
\end{align}

In our example, it appears that the real part of $h_{\mathrm{c}}$ measures the phase-space uncertainty of the state when all the quantum interferences have been lost (in which case the state is quasiclassical, \textit{i.e.}, a mixture of coherent states).
On the other hand, the imaginary part of $h_{\mathrm{c}}$ measures the average number of almost distinguishable quasiclassical components that are interfering with any
quasiclassical
component of the superposition state.

\bigskip

\end{document}